\documentclass[onefignum,onetabnum]{siamart171218}

\usepackage{lipsum}
\usepackage{amsfonts}
\usepackage{graphicx}
\usepackage{epstopdf}
\usepackage{algorithmic}
\ifpdf
  \DeclareGraphicsExtensions{.eps,.pdf,.png,.jpg}
\else
  \DeclareGraphicsExtensions{.eps}
\fi

\newsiamremark{remark}{Remark}
\newsiamremark{hypothesis}{Hypothesis}
\crefname{hypothesis}{Hypothesis}{Hypotheses}
\newsiamthm{claim}{Claim}

\headers{Singular DMD}{J. A. Rosenfeld and R. Kamalapurkar}

\title{Singular Dynamic Mode Decompositions\thanks{A YouTube playlist to accompany this work may be found at: \url{https://youtube.com/playlist?list=PLldiDnQu2phsZdFP3nHoGnk_Aq-kp_4nE}
\funding{This research was supported by the Air Force Office of Scientific Research (AFOSR) under contract numbers FA9550-20-1-0127 and FA9550-21-1-0134, and the National Science Foundation (NSF) under award numbers 2027976 and 1900364. Any opinions, findings and conclusions or recommendations expressed in this material are those of the author(s) and do not necessarily reflect the views of the sponsoring agencies.}}}

\author{Joel A. Rosenfeld%
\thanks{Department of Mathematics and Statistics, University of South Florida, Tampa, FL 33620 USA (\email{rosenfeldj@usf.edu}, \url{http://thelearningdock.org})}%
\and Rushikesh Kamalapurkar%
\thanks{Department of Mechanical and Aerospace Engineering, Oklahoma State University, Stillwater, OK 74078 USA (\email{rushikesh.kamalapurkar@okstate.edu}, \url{https://scc-lab.github.io})}}%

\usepackage{amsopn}
\DeclareMathOperator{\diag}{diag}

\makeatletter
\newcommand*{\addFileDependency}[1]{
  \typeout{(#1)}
  \@addtofilelist{#1}
  \IfFileExists{#1}{}{\typeout{No file #1.}}
}
\makeatother

\newcommand*{\myexternaldocument}[1]{%
    \externaldocument{#1}%
    \addFileDependency{#1.tex}%
    \addFileDependency{#1.aux}%
}

\ifpdf
\hypersetup{
  pdftitle={Singular Dynamic Mode Decomposition},
  pdfauthor={J. A. Rosenfeld, R. Kamalapurkar}
}
\fi

\myexternaldocument{ex_supplement}

\usepackage{savesym}

\usepackage{graphicx}
\usepackage{amssymb,amsfonts,amsmath}
\usepackage[caption=false]{subfig}
\usepackage[normalem]{ulem}
\useunder{\uline}{\ul}{}

\usepackage[normalem]{ulem}
\useunder{\uline}{\ul}{}
\usepackage{algorithmic}

\newtheorem{example}{Example}

\usepackage{color}

\begin{document}

\maketitle

\begin{abstract}
  This manuscript is aimed at addressing several long standing limitations of dynamic mode decompositions in the application of Koopman analysis. Principle among these limitations are the convergence of associated Dynamic Mode Decomposition algorithms and the existence of Koopman modes. To address these limitations, two major modifications are made, where Koopman operators are removed from the analysis in light of Liouville operators (known as Koopman generators in special cases), and these operators are shown to be compact for certain pairs of Hilbert spaces selected separately as the domain and range of the operator. While eigenfunctions are discarded in the general analysis, a viable reconstruction algorithm is still demonstrated, and the sacrifice of eigenfunctions realizes the theoretical goals of DMD analysis that have yet to be achieved in other contexts. However, in the case where the domain is embedded in the range, an eigenfunction approach is still achievable, where a more typical DMD routine is established, but that leverages a finite rank representation that converges in norm. The manuscript concludes with the description of two Dynamic Mode Decomposition algorithms that converges when a dense collection of occupation kernels, arising from the data, are leveraged in the analysis.
\end{abstract}



\section{Introduction}

This manuscript is aimed at addressing several long standing limitations of dynamic mode decompositions (DMD) in the application of Koopman analysis. Principle among these limitations are the convergence of associated Dynamic Mode Decomposition algorithms and the existence of Koopman modes, where the first has only been established with respect to the strong operator topology (which does not guarantee the convergence of the spectrum), and the second is only guaranteed to exist when the Koopman operator is compact as well as self-adjoint or normal (which is a rare occurrence over the typical sample spaces).

DMD methods are data analysis methods that aim to decompose a time series corresponding to a nonlinear dynamical system into a collection of dynamic modes \cite{kutz2016dynamic,budivsic2012applied,mezic2005spectral,korda2018convergence}. When they are effective, a given time series can be expressed as a linear combination of dynamic modes and exponential functions whose growth rates are derived from the spectrum of a finite rank representation of a particular operator, usually the Koopman operator.


The use of Koopman operators places certain constraints on the continuous time dynamics that can be studied with DMD methods. In particular, Koopman operators analyze continuous time dynamics through a discrete time proxy obtained by fixing a time-step for a continuous time system \cite{mauroy2016global}. However, only a small subset of continuous time dynamics satisfy the forward invariance property necessary to obtain a discretization  \cite{rosenfeld2019dynamic}. Moreover, to establish convergence guarantees for DMD routines, additional structure is required of Koopman operators, where a sequence of finite rank operators converge to Koopman operators in norm only if the Koopman operator is compact \cite{pedersen2012analysis}. Compactness is rarely satisfied for Koopman operators, where the Koopman operators obtained through discretizations of the simplest dynamical system $\dot x = 0$ are the identity operator and are not compact. A partial result has been demonstrated for when Koopman operators are bounded in \cite{korda2018convergence}, where a sequence of finite rank operators converge to a Koopman operator in the Strong Operator Topology (SOT). However, SOT convergence does not guarantee convergence of the spectra (cf. \cite{pedersen2012analysis}), which is necessary for a DMD routine.

There are stronger theoretical difficulties associated with Koopman operators. It has been demonstrated that among the typical Hilbert spaces leveraged in sampling theory, such as the exponential dot product's \cite{carswell2003composition}, the Gaussian RBF's \cite{gonzalez2021anti}, and the polynomial kernel's native spaces as well as the classical Paley Wiener space \cite{chacon2007composition}, the only discrete time dynamics that yield a bounded Koopman operator are those dynamics that are affine. Hence, depending on the kernel function selected for the approximation of a Koopman operator, a given Koopman operator can at best be expected to be a densely defined operator, which obviates the aforementioned convergence properties.

Another motivation for the use of Koopman operators in the study of continuous time dynamical systems is a heuristic that for small timesteps the spectra and eigenfunctions of the resultant Koopman operator should be close to that of the Liouville operator representing the continuous time systems \cite{brunton2019data}. However, for two fixed timesteps, the corresponding Koopman operators can have different collections of eigenfunctions and eigenvalues, and these are artifacts of the discretization itself \cite{gonzalez2021anti}. Since in most cases the Koopman operators are used for this analysis, it is not clear if there is a method for distinguishing which of these eigenfunctions and eigenvalues are a product of the discretization and which are fundamental to the dynamics themselves.

Finally, and perhaps most alarming, is that Koopman modes themselves exist for only a small subset of Koopman operators \cite{gonzalez2021anti}. Specifically, if a Koopman operator is self-adjoint, then it admits an orthonormal basis of eigenfunctions \cite{brunton2019data}, and the projection of the full state observable onto this basis yields a collection of (vector valued) coefficients attached to these basis functions. These coefficients are known as Koopman Modes or Dynamic Modes. Koopman operators are not necessarily diagonalizable over a given Hilbert space, and when they are diagonalizable, their complete eigenbasis is not always an orthogonal basis. Hence, as the full state observable is projected on larger and larger finite collections of eigenfunctions, the weights attached to each eigenfunction will change as more are added. This adjustment to the weights with the addition of more eigenfunctions is why a series expansion is only ever given in Hilbert space theory when there is an orthonormal basis of eigenfunctions, otherwise an expansion is written as limit of finite linear combinations of eigenfunctions \cite{pedersen2012analysis}.\footnote{There are notable exceptions, such as in atomic decompositions \cite{zhu2012analysis}. However, this is another rare property of a basis.}

To address these limitations, two major modifications are made, where Koopman operators are removed from the analysis in light of Liouville operators (known as Koopman generators in special cases), and these operators are shown to be compact for certain pairs of Hilbert spaces selected separately as the domain and range of the operator. (This separation of the domain and range is not possible for Koopman operators.) While eigenfunctions are discarded in the general analysis, a viable reconstruction algorithm is still achievable, and the sacrifice of eigenfunctions realizes the theoretical goals of DMD analysis that have yet to be achieved in other contexts. It should be noted that Liouville and Koopman operators rarely admit a diagonalization, and as such, this approach discards that additional assumption on the operators.

However, at the cost of well defined Dynamic Modes, an eigenfunction approach is still achievable when the domain is embedded in the range of the operator. This allows for the search of eigenfunctions through finite rank approximations that converge to the Liouville operator. The result is a norm convergence DMD routine (using eigenfunctions), which is an achievement over the SOT convergent results previously established in the field \cite{korda2018convergence}. This gives a balance between the two convergence methods presented in this manuscript, where well defined modes come at the price of ease of reconstruction, and a straightforward reconstruction algorithm may not have well defined limiting dynamic modes (a problem shared with all other DMD routines).

To be explicit, the singular DMD approach yields the following benefits:
\begin{enumerate}
    \item Eliminates the requirement of forward invariance. (Aligning with the method given in \cite{rosenfeld2019occupation}).
    \item Provides well defined Dynamic Modes.
    \item Approximates a compact operator, thereby achieving convergence.
    \item Yields an orthonormal basis through which the full state observable may be decomposed.
\end{enumerate}

However, this achievement comes at the expense of eigenfunctions of the given operator. As it turns out, the abandonment of eigenfunctions for the analysis does not actually limit the applicability, where even for very simple dynamics, such as $f(x) = x^2$ in the one dimensional setting, the corresponding Liouville operators will have no eigenfunctions over any space of continuous functions. For the present example, the solution to the eigenfunction equation, $g'(x)x^2 = \lambda g(x)$, gives $g(x) = e^{\lambda/x}$ for $\lambda \neq 0$, a discontinuous function on the real line. Additionally, reconstruction of the original time series may still be achieved using Runge-Kutta like methods.

Where the DMD routine leveraging the case where the domain is embedded in the range provides the following:
\begin{enumerate}
    \item Eliminates the requirement of forward invariance. (Aligning with the method given in \cite{rosenfeld2019occupation}).
    \item Approximates a compact operator, thereby achieving convergence.
    \item Yields an approximate eigenbasis through which the full state observable may be decomposed.
    \item An ease of reconstruction through the eigenfunctions.
\end{enumerate}

It should be noted that there have been several attempts at providing compact operators for the study of DMD. The approaches \cite{das2021reproducing} and \cite{rosenfeld2019dynamic} find compact operators through the multiplication of auxiliary operator against Koopman and Liouville operators respectively. However, the resultant operators are not the operators that truly correspond to the dynamics in question, and as such, the decomposition of those operators can only achieve heuristic results. The approach taken presently gives compact Liouville operators directly connected with the continuous time dynamics.



\section{Reproducing Kernel Hilbert Spaces}

A reproducing kernel Hilbert space (RKHS), $H$, over a set $X$ is a space of functions from $X$ to $\mathbb{R}$ such that the functional of evaluation, $E_x g := g(x)$ is bounded for every $x \in X$. By the Riesz theorem, this means for each $x \in X$ there exists a function $K_x \in H$ such that $\langle f, K_x \rangle_H = f(x)$ for all $f$. The function $K_x$ is called the kernel function centered at $X$, and the function $K(x,y) := \langle K_y, K_x \rangle_H$ is called the kernel function corresponding to $H$. Note that $K_y(x) = K(x,y).$ Classical examples of kernel functions in data science are the Gaussian radial basis function for $\mu > 0$, $K(x,y) = \exp(-\frac{1}{\mu} \| x- y\|^2)$, and the exponential dot product kernel, $\exp(\frac{1}{\mu} x^T y)$ \cite{steinwart2008support}.

The function $K(x,y)$ is a positive definite kernel function, which means that for every finite collection of points, $\{ x_1, \ldots, x_M \} \subset X$, the Gram matrix $( K(x_i,x_j) )_{i,j=1}^M$ is positive definite. For each positive definite kernel function, there exists a unique RKHS for which $K$ is the kernel function for that space by the Aronszajn-Moore theorem in \cite{aronszajn1950theory}.

Given a RKHS, $H$, over $X \in \mathbb{R}^n$ consisting of continuous functions and given a continuous signal, $\theta:[0,T] \to X$, the linear functional $g \mapsto \int_0^T g(\theta(t)) dt$ is bounded. Hence, there exist a function, $\Gamma_{\theta} \in H$, such that $\langle g, \Gamma_{\theta} \rangle_{H} = \int_0^T g(\theta(t)) dt$ for all $g \in H$. The function $\Gamma_{\theta}$ is called the occupation kernel in $H$ corresponding to $\theta$. These occupation kernels were first introduced in \cite{rosenfeld2019occupation,rosenfeld2019occupation2}.

\section{Compact Liouville Operators}
This section demonstrates the existence of compact Liouville operators, given formally as $A_f g(x) = \nabla g(x) f(x)$, where compactness is achieved through the consideration of differing spaces for the domain an range of the operator. Section \ref{sec:classical} builds on a classical result where differentiation between differing weighted Hardy spaces can be readily shown to be compact. Following a similar argument, Section \ref{sec:severalvariables} presents several examples of compact Liouville operators over spaces of functions of several variables. We would like to emphasize that the collections of compact Liouville operators are not restricted to these particular pairs of functions spaces, but rather this section provides several examples demonstrating the existence of such operators, thereby validating the approach in the sequel.

\subsection{Inspirations from Classical Function Theory}\label{sec:classical}
Consider the weighted Hardy spaces (cf. \cite{beneteau2018remarks}), $H^2_{\omega}$, where $\omega = \{ \omega_{m} \}_{m=0}^\infty$ is a sequence of positive real numbers such that $|\omega_{m+1}/\omega_m| \to 1$, and $g(z) = \sum_{m=0}^\infty a_m z^m$ is a function in $H^2_{\omega}$ if the coefficients of $g$ satisfy $\|g\|_{H_{\omega}^2}^2 := \sum_{m=0}^\infty \omega_m |a_m|^2 < \infty$. Each weighted Hardy space is a RKHS over the complex unit disc $\mathbb{D} =\{ z \in \mathbb{C} : |z| = 1\}$ with kernel function given as $K_{\omega}(z,w) = \sum_{m=0}^\infty \omega_m z^m \bar w^m$, and the monomials $\left\{ \frac{z^m}{\sqrt{\omega_m}} \right\}_{m=0}^\infty$ form an orthonormal basis for each space.

The weighted Hardy space corresponding to the sequence $\omega_{(0)} := \{ 1, 1, \ldots \}$ is the classical Hardy space, $H^2$, that was introduced by Riesz in 1923 \cite{riesz1923randwerte}. The Dirichlet space corresponds to the weight sequence $\omega_{(1)} = \{ (m+1) \}_{m=0}^\infty$, and the Bergman space corresponds to $\omega_{(-1)} = \{ (m+1)^{-1} \}_{m=0}^\infty$. Of interest here is the weighted Hardy space corresponding to $\omega_{(3)} := \{ m^3 \}_{m=0}^\infty$, which will be denoted as $H^{2}_3$ for convenience.

It is immediately evident that the operation of differentiation on elements of $H^2_3$ is bounded as an operator from $H^2_3$ to $H^2$. The reason for this inclusion can be seen directly through the power series for these function spaces. In particular, a function $h(z) = \sum_{m=0}^\infty a_m z^m$ is in $H^2_2$ if $\| h\|_{H_2^2} = \sum_{m=0}^\infty (m+1)^3 |a_m|^2 < \infty$, and in the Hardy space if $\| h\|_{H^2} = \sum_{m=0}^\infty  |a_m|^2 < \infty.$

A function $g$ in $H^2_3$ has derivative $g'(z) = \sum_{m=1}^\infty m a_m z^{m-1} = \sum_{m=0} (m+1) a_{m+1} z^{m}$, and by considering the Hardy space norm,
\[\left\| \frac{d}{dz} g \right\|_{H^2} = \sum_{m=0}^\infty (m+1)^2 |a_{m+1}|^2 \le \sum_{m=0}^\infty (m+1)^3 |a_{m+1}|^2,\]
but this is exactly the $H^2_3$ norm on $g$ less the constant term. Hence differentiation is a bounded operator from the space $H^2_3$ to the Hardy space with operator norm at most $1$.

\begin{proposition}The operator $\frac{d}{dz} : H_3^2 \to H^2$ is compact. Moreover, if $f:\overline{\mathbb{D}} \to \mathbb{D}$ is a bounded analytic function corresponding to a bounded multiplication operator, $M_f g := g(x) f(x)$, over the Hardy space, then the Liouville operator, $A_f := M_f\frac{d}{dz}$, is compact from $H_3^2$ to $H^2$.\end{proposition}
\begin{proof}
To see that differentiation is a compact operator from the $H^2_3$ to the Hardy space, we may select a sequence of finite rank operators that converge in norm to differentiation. In particular, note that the monomials form an orthonormal basis of the Hardy space as is evident from the given norm. Let $\alpha_M := \{ 1, z, \ldots, z^M\}$ be the first $M$ monomials in $z$, and let $P_{\alpha_M}$ be the projection onto the span of these monomials. The operator $P_{\alpha_M} \frac{d}{dz}$ is a finite rank operator, where the image of this operator is a polynomial of degree up to $M$.

To demonstrate that this sequence of finite rank operators converges to differentiation in the operator norm it must be shown that the difference under the operator norm, \[\left\| P_{\alpha_M} \frac{d}{dz} - \frac{d}{dz} \right\|_{H_3^2}^{H^2} := \sup_{g \in H_3^2} \frac{\| P_{\alpha_M} \frac{d}{dz} g - \frac{d}{dz} g\|_{H^2}}{\|g\|_{H_3^2}},\]
goes to zero. Note that
\begin{gather*}\| P_{\alpha_M} \frac{d}{dz} g - \frac{d}{dz} g\|_{H^2}^2 = \sum_{m=M+1}^\infty (m+1)^2 |a_{m+1}|^2 \\= \sum_{m=M+1}^\infty \frac{1}{m+1} (m+1)^3 |a_{m+1}|^2 \le \frac{1}{M+1} \sum_{m=M+1}^\infty (m+1)^3 |a_{m+1}|^2 \le \frac{1}{M+1} \|g\|_{H_{3}^2}.\end{gather*} Hence $\left\| P_{\alpha_M} \frac{d}{dz} - \frac{d}{dz} \right\|_{H_3^2}^{H^2} \le \frac{1}{M+1} \to 0.$ This proves that differentiation is a compact operator from $H_3^2$ to $H^2$.

If a function, $f$, is a bounded analytic function on the closed unit disc, then it is the symbol for a bounded multiplier over $H^2$. Hence, the $M_f \frac{d}{dz}$ is a compact operator from $H_3^2$ to $H^2$. To be explicit, since $P_{\alpha_M} \frac{d}{dz}$ has finite rank, $M_f \left(P_{\alpha_M} \frac{d}{dz}\right)$ also has finite rank. Moreover, $\left\| M_f P_{\alpha_M} \frac{d}{dz} - M_f \frac{d}{dz} \right\|_{H_3^2}^{H^2} = \left\| M_f \left(P_{\alpha_M} \frac{d}{dz} - \frac{d}{dz}\right) \right\|_{H_3^2}^{H^2} \le \| M_f \|_{H^2}^{H^2} \left\| P_{\alpha_M} \frac{d}{dz} - \frac{d}{dz} \right\|_{H_3^2}^{H^2}  \to 0.$ Hence, $M_f \frac{d}{dz}$ is an operator norm limit of finite rank operators, and is compact. Finally, it can be seen that $M_f \frac{d}{dz} g(z) = g'(z) f(z) = A_f g(z)$, and $A_f$ is a compact Liouville operator from $H_3^2$ to $H^2$.
\end{proof}

\subsection{Compact Liouville Operators of Several Variables}\label{sec:severalvariables}

The example of the previous section demonstrated that compact Liouville operators may be obtained in one dimension. However, this is readily extended to higher dimensions through similar arguments, and in particular can be demonstrated for dot product kernels of the form $K(x,y) = (1+\mu x^T y)^{-1}$. In some cases, such as with the exponential dot product kernel and the Gaussian RBF, where the kernel functions over $\mathbb{R}^n$ decompose as a product of kernel functions over $\mathbb{R}$ for the individual variables, the establishment of compact Liouville operators from the single variable spaces to an auxiliary range RKHSs yields compact Liouville operators through tensor products of the respective spaces.

The exponential dot product kernel, with parameter $\mu > 0$, is given as $K(x,y) = exp\left(\mu x^Ty\right)$. In the single variable case, the native space for this kernel may be expressed as $F^2_{\mu}(\mathbb{R}^n) = \left\{ f(x) = \sum_{m=0}^\infty a_m x^m : \sum_{m=0}^\infty |a_m|^2 \frac{m!}{\mu^m} < \infty \right\}$. This definition can be readily extended to higher dimensions, where collection of monomials, $x^{\alpha} \frac{\mu^{|\alpha|}}{\sqrt{\alpha!}}$, with multi-indices $\alpha \in \mathbb{N}^n$ form an orthonormal basis. The norm of functions in $F_\mu^2(\mathbb{R}^n)$ will be denoted by $\|g\|_\mu.$

In this setting, if $\mu_2 > \mu_1$ (i.e. $1/\mu_1 > 1/\mu_2$), then by arguments similar to those given in the previous section, it follows that partial differentiation with respect to each variable is a compact operator from $F^2_{\mu_1}$ to $F^2_{\mu_2}$. However, since multiplication operators are unbounded from $F^2_{\mu}$ to itself for every $\mu > 0$, another step is necessary to ensure compactness.

\begin{lemma}
Suppose that $\eta < \mu$, then given any polynomial of several variables, $f$, the multiplication operator $M_{f} : F_\eta^2(\mathbb{R}^n) \to F_\mu^2(\mathbb{R}^n)$ is bounded.
\end{lemma}
\begin{proof}To facilitate a clarity of exposition, this will be proven with respect to functions of a single variable. The same arguments extend to the spaces of several variables, albeit with more bookkeeping.

Let $g \in F^2_{\eta}$. Then $g(x) = \sum_{m=0}^\infty a_m x^m$, and $\| g\|_\eta^2 = \sum_{m=0}^\infty |a_m|^2 \frac{m!}{\eta^m}$.

For $f \equiv 1$, $M_1$ is the identity operator. Thus, the boundedness of $M_1$ is equivalent to demonstrating that $F_\eta^2$ is boundedly included in $F_\mu^2$. In particular, note that 
\begin{gather*}
    \|M_1 g \|_\mu^2 = \| g \|_\mu^2 = \sum_{m=0}^\infty |a_m|^2 \frac{m!}{\mu^m} = \sum_{m=0}^\infty |a_m|^2 \left(\frac{\eta}{\mu}\right)^m \frac{m!}{\eta^m} \\< \sum_{m=0}^\infty |a_m|^2 \frac{m!}{\eta^m} = \| g \|^2_\eta
\end{gather*}

Fix $k \in \mathbb{N}$ and consider the multiplication operator $M_{x^k} : F^2_\eta \to F^2_\mu$ defined as $M_{x^k} g := xg$ for all $g \in F^2_{\eta}$. Note that the power series of $M_{x^k} g$ is given as $xg(x) = \sum_{m=0}^\infty a_m x^{m+k} = \sum_{m=k}^\infty a_{m-k} x^m$. Hence,
\begin{gather*}\| x^k g(x) \|^2_{\mu} = \sum_{m=k}^\infty |a_{m-k}|^2 \frac{m!}{\mu^m}= \sum_{m=0}^\infty |a_{m}|^2 \frac{(m+k)!}{\mu^{m+k}}\\ = \sum_{m=0}^\infty |a_m|^2 \frac{(m+k)!}{m!\mu^k} \frac{m!}{\mu^m} = \sum_{m=0}^\infty |a_m|^2 \left(\frac{m+k}{m!\mu^k}\right) \left(\frac{\eta}{\mu}\right)^m  \frac{m!}{\eta^m},
\end{gather*}
and as $\left(\frac{m+k}{m!\mu^k}\right) \left(\frac{\eta}{\mu}\right)^m$ is bounded as a function of $m$ by some constant $C > 0$ (owing to the exponential decay of $\left(\eta/\mu\right)^m$), it follows that $\| M_{x^k} \|_{F_\eta^2}^{F_\mu^2} < C$.

Hence, by linear combinations of monomials it has been demonstrated that a multiplication operator with polynomial symbol is a bounded operator.
\end{proof}

\begin{remark}
The authors emphasize that the collection of bounded multiplication operators between these spaces is strictly larger than the those with polynomial symbols. The purpose of this lemma is to simply support the existence of compact Liouville operators, rather than to provide a complete classification.
\end{remark}

\begin{theorem}
Let $\mu_3 > \mu_1$, and suppose that $f$ is a vector valued function over several variables, where each entry is a polynomial. Then the Liouville operator $A_f : F^2_{\mu_1}(\mathbb{R}^n) \to F^2_{\mu_3}(\mathbb{R}^n)$ defined as $A_f g = \nabla g \cdot f$ is a compact operator. 
\end{theorem}

\begin{proof}
Let $f = (f_1, f_2, \ldots, f_n)^T$, and select $\mu_2$ such that $\mu_1 < \mu_2 < \mu_3$. For each $i = 1,\ldots, n$, the operator of partial differentiation $\frac{\partial}{\partial x_i} : F_{\mu_1}^2(\mathbb{R}^n) \to F_{\mu_2}^2(\mathbb{R}^n)$ is a compact operator, and the multiplication operator $M_{f_i} : F_{\mu_2}^2(\mathbb{R}^n) \to F_{\mu_3}^2(\mathbb{R}^n)$ is bounded. Hence, the operator $M_{f_i} \frac{\partial}{\partial x_i}$ is compact. As $A_f = M_{f_1} \frac{\partial}{\partial x_1} + \cdots + M_{f_n} \frac{\partial}{\partial x_n}$, it follows that $A_f$ is a compact operator from $F_{\mu_1}^2(\mathbb{R}^n)$ to $F_{\mu_3}^2(\mathbb{R}^n)$.
\end{proof}

This section has thus established the existence of compact Liouville operators between various pairs of spaces. It is emphasized that these are not the only pairs for which a compact Liouville may be determined.

\section{Singular Dynamic Mode Decompositions for Compact Liouville Operators}

The objective of this section is to determine a decomposition of the full state observable, $g_{id}(x) := x$, with respect to an orthonormal basis obtained from a Liouville operator corresponding to a continuous time dynamical system $\dot x = f$. We will let $H$ and $\tilde H$ be two RKHSs over $\mathbb{R}^n$ such that the Liouville operator, $A_f g(x) = \nabla g(x) f(x)$ is compact as an operator from $H$ to $\tilde H$. 
To obtain an orthonormal basis, a singular value decomposition for the compact operator $A_f$ is obtained. Specifically, note that as $A_f$ is compact, so is $A_f^*$. Hence, $A_f^* A_f$ is diagonalizable as a self adjoint compact operator. Thus, there is a countable collection of non-negative eigenvalues $\sigma_m^2 \ge 0$ and eigenfunctions $\varphi_m$ corresponding to $A_f^* A_f$, such that $A_f^* A_f \varphi_m = \sigma_m^2 \varphi_m$. Since $A_f^* A_f$ is self adjoint, $\{ \varphi_m \}_{m=0}^\infty$ may be selected in such a way that they form an orthonormal basis of $H$. The functions $\varphi_m$ are the right singular vectors of $A_f$.

For $\sigma_m \neq 0$, the left singular vectors may be determined as $\psi_m := \frac{A_f \varphi_m}{\sigma_m},$ and the collection of nonzero $\psi_m$ form an orthonormal set in $\tilde H$. This may be seen via 
\begin{gather*}\langle \psi_m, \psi_{m'} \rangle_{\tilde H} = \frac{1}{\sigma_{m}\sigma_{m'}}\langle A_f \varphi_m, A_f \varphi_{m'} \rangle_{\tilde H}\\= \frac{1}{\sigma_{m}\sigma_{m'}} \langle \varphi_m, A_f^* A_f \varphi_{m'} \rangle_{H} = \frac{\sigma_{m'}^2}{\sigma_{m}\sigma_{m'}} \langle \varphi_m, \varphi_{m'} \rangle= \frac{\sigma_{m'}^2}{\sigma_{m}\sigma_{m'}} \delta_{m,m'},\end{gather*} where $\delta_{\cdot,\cdot}$ is the Kronecker delta function.

Finally, \[A_f g = \sum_{\sigma_m \neq 0} \langle g, \varphi_m \rangle_H \sigma_m \psi_m\] for all $g \in H$, and \[A_f^* h = \sum_{\sigma_m \neq 0} \langle h, \psi_m \rangle_{\tilde H} \sigma_m \varphi_m.\]

To find a decomposition for the full state observable, $g_{id}$, first note that the full state observable is vector valued, whereas the Hilbert spaces consist of scalar valued functions. To ameliorate this discrepancy, we will work with the individual entries of the full state observable, namely the maps $x \mapsto (x)_i$, for $i=1,\ldots,n$, which are the mappings of $x$ to its individual components. When $(x)_i$ resides in the Hilbert space, such as with the space $F_\mu^2(\mathbb{R}^n)$, and $(x)_i$ may be directly expanded with respect to the right singular vectors of $A_f$. If $(x)_i$ is not in the space, as in the case with the Gaussian RBF, if the space is universal, then a suitable approximation may be determined over a fixed compact subset, and the approximation will be expanded instead. Performing the entry wise decomposition of the full state observable is equivalent to performing the decomposition over vector valued RKHSs with diagonal kernel operators, and replacing the gradient of $g$ with the matrix valued derivative.

Hence, for each $i=1,\ldots,n$, we have $(x)_i = \sum_{m=0}^\infty (\xi_m)_i \varphi_m(x)$, where $(\xi_m)_i = \langle (x)_i, \varphi_m \rangle_{H}$. The vectors $\xi_m$ are called the \emph{singular Liouville modes} of the dynamical system with respect to the pair of Hilbert space $H$ and $\tilde H$.

Note that for a trajectory of the system, given as $x(t)$, it can be seen that
\begin{gather*}
    \dot x(t) = f(x(t)) = \nabla g_{id}(x(t)) f(x(t)) = A_f g_{id}(x(t))\\
    = \sum_{m=0}^\infty \langle g_{id}, \varphi_{m} \rangle_{H} \sigma_m \psi_m(x(t)) = \sum_{m=0}^\infty \xi_m \sigma_m \psi_m(x(t)).
\end{gather*}

Hence, $x(t)$ satisfies a differential equation with respect to the left singular vectors of the Liouville operator and the singular Liouville modes. Given these quantities, reconstruction of $x(t)$ is possible using tools from the solution of initial value problems. In particular, the following form of the equation may be exploited:
\[ x(t) = x(0) + \sum_{m=0}^\infty \xi_m \sigma_m \int_0^t \psi_m(x(\tau)) d\tau. \]

\section{Recovering an Eigenfunction Approach in Special Cases}

While the majority of this mansucript is aimed at the singular Dynamic Mode Decomposition, where the domain and range are different for the compact Liouville operator, there is still a possibility of obtaining an eigendecomposition in special cases. In particular, for many of the examples shown above, the domain and range spaces have similar structure and the range space has less stringent requirement for the functions it contains. This means that the domain itself may be embedded in the range space, and if there is a complete set of eigenfunctions in this embedded space, then the operator may still be diagonalized.

Note that the operator is still mapping between two different Hilbert spaces, which means that the inner product on the embedding is different than the inner product on the domain. This difference will appear in the numerical methods given in subsequent sections.

The following is a well known result (cf. \cite{zhu2012analysis}), and is included here for illustration purposes.

\begin{proposition}
If $\mu_1 < \mu_2$, then $F^2_{\mu_1}(\mathbb{R}^n) \subset F^2_{\mu_2}(\mathbb{R}^n)$.
\end{proposition}

\begin{proof}
Again this is shown for the single variable case, where the multivariate case follows by an identical argument, but with more bookkeeping.

Suppose that $g \in F^2_{\mu_1}(\mathbb{R})$ with $g(z) = \sum_{m=0}^\infty a_m z^m$. Then \begin{gather*}\|g\|_{F_{\mu_2}^2(\mathbb{R})}^2 = \sum_{m=0}^\infty |a_m|^2 \frac{m!}{\mu_2^m} = \sum_{m=0}^\infty |a_m|^2 \left( \frac{\mu_1}{\mu_2} \right)^m \frac{m!}{\mu_1^m} \le \sum_{m=0}^\infty |a_m|^2 \frac{m!}{\mu_1^m} = \|g\|_{F_{\mu_1}^2(\mathbb{R})}^2.\end{gather*}

Since the quantity on the right is bounded, so is the quantity on the left. Hence $g \in F_{\mu_2}^2(\mathbb{R}).$
\end{proof}

\begin{example}
A simple example demonstrating that an eigenbasis may be found between the two spaces arises in the study of $A_x:F_{\mu_1}^2(\mathbb{R}) \to F_{\mu_2}^2(\mathbb{R})$ for $\mu_1 < \mu_2$. Note that an eigenfunction, $\varphi$, for $A_z$ must reside in $F_{\mu_1}^2(\mathbb{R}) \cap F_{\mu_2}^2(\mathbb{R}) = F_{\mu_1}^2(\mathbb{R})$, and satisfy $\varphi'(x) x = \lambda \varphi(x)$. Consequently, takes the form $\varphi(x) = x^{\lambda}$, and is in $F_{\mu_1}^2(\mathbb{R})$ only for $\lambda = 0, 1, 2, \ldots$. Hence, the eigenfuncitons of $A_x$ are the monomials. Monomials are contained in $F_{\mu_1}^2(\mathbb{R})$ and form a complete eigenbasis for both spaces. Note that the norm of $x^m$ is $\sqrt{\frac{m!}{\mu_1^m}}$ in $F_{\mu_{1}}^2(\mathbb{R})$ and $\sqrt{\frac{m!}{\mu_2^m}}$ in $F_{\mu_{2}}^2(\mathbb{R}).$
\end{example}

The following proposition is obtained in the same manner as in the classical case.

\begin{proposition}
Suppose that $H$ and $\tilde H$ are two RKHSs over $\mathbb{R}^n$, and that $H \subset \tilde H$. If $\varphi \in H$ is an eigenfunction for $A_f$ as $A_{f} \phi = \lambda \phi$, then given a trajectory $x:[0,T] \to \mathbb{R}^n$ satisfying $\dot x = f(x)$ the following holds $\varphi(x(t)) = e^{\lambda t} \varphi(x(0)).$
\end{proposition}

\begin{proof}
Since $A_{f}\varphi = \nabla \varphi f$, it follows that \[\frac{d}{dt} \varphi(x(t)) = \nabla \varphi(x(t)) \dot x(t) = \nabla \varphi(x(t)) f(x(t)) = A_f \varphi(x(t)) = \lambda \varphi(x(t)).\] That is, $\frac{d}{dt} \varphi(x(t)) = \lambda \varphi(x(t)).$ Thus, the conclusion follows.
\end{proof}

Suppose that $A_f : H \to \tilde H$ has a complete eigenbasis in the sense that the span of the eigenfunctions, $\{ \varphi_m \}_{m = 1}^\infty$, are dense in $H$. Then the full state observable, $g_{id}$, is the full state observable, then each entry of $g_{id}$, $(x)_i$ for $i=1,\ldots,n$, may be expressed as
\[ (x)_i = \lim_{M\to\infty} \sum_{m=1}^M (\xi_{m,M})_i \varphi_m(x),\]
where $(\xi_{m,M})_i$ is the $m$-th coefficient obtained from projecting $(x)_i$ onto the span of the first $M$ eigenfunctions. If the eigenfunctions are orthogonal, then the dependence on $M$ may be removed from $\xi_{m,M}$. Hence, the full state observable is obtained from
\begin{equation}\label{eq:eigendecomp} g_{id}(x) = \lim_{M\to\infty} \sum_{m=1}^M \xi_{m,M} \varphi_m(x),\end{equation}
with $\xi_{m,M}$ being the vector obtained by stacking $(\xi_{m,M})_i$. Finally, by substituting $x(t)$ into this representation (where $\dot x = f(x)$), the following holds
\begin{equation}\label{eq:dmdrepresentation} x(t) = g_{id}(x(t)) = \lim_{M\to\infty} \sum_{m=1}^M \xi_{m,M} e^{\lambda t} \varphi_m(x(0)).\end{equation}

Hence, this methodology yields a DMD routine, where the finite rank representations will converge to the compact Liouville operators, following the proof given in the Appendix of \cite{rosenfeld2019dynamic}.

\section{Singular Dynamic Mode Decomposition Algorithm}
This section is aimed at determining a convergent algorithm that can determine approximations of the singular Liouville modes and the singular vectors of $A_f$. While an eigenfunction expansion is still possible in the case of nested spaces, the Singular Dynamic Mode Decomposition algorithm is technically more general. Moreover, the SVD ensures the existence of dynamic modes, which may not be well defined fixed concepts for the eigenfunction case.

From the data perspective, a collection of trajectories, $\{ \gamma_j : [0,T_j] \to \mathbb{R}^n \}_{j=1}^M$, corresponding to an unknown dynamical system, $f:\mathbb{R}^n \to \mathbb{R}^n$, as $\dot \gamma_j = f(\gamma_j)$ have been observed. The objective of DMD is to get an approximation of the dynamic modes of the system, and to obtain an approximate reconstruction of a given trajectory. Once a reconstruction is determined, then data driven predictions concerning future states of the trajectory may be determined. A DMD routine is somewhat like a Fourier series representation, which can reproduce a continuous trajectory exactly, however DMD methods exploit a trajectory's underlying dynamic structure.

This routine effectively interpolates the action of the Liouville operator on a collection of basis functions. When these basis functions form a complete set within the Hilbert space, which can be achieved by selecting a dense collection of short trajectories throughout the workspace, then a sequence of finite rank approximations determined by this routine converges to the compact Liouville operator in norm. Which means that the left and right singular functions of the finite rank operators in the sequence converge to those of the Liouville operator, and that the singular values converge as well.

DMD routines involving the Koopman operator add the additional requirement of forward invariance for the sake of discretizations. This method as well as that of \cite{rosenfeld2019dynamic} sidestep that requirement by accessing the Liouville operators directly through their connection with the occupation kernels of the RKHSs. To wit, given two RKHSs of continuously differentiable functions, $H$ and $\tilde H$, with kernels $K$ and $\tilde K$ respectively, and a compact Liouville operator, $A_f : H \to \tilde H$, the occupation kernel, $\Gamma_{\gamma_{j}} \in \tilde H$ corresponding to the trajectory $\gamma_j$ satisfies $A_f^* \Gamma_{\gamma_j} = K(\cdot,\gamma_j(T_j)) - K(\cdot,\gamma_j(0)),$ where $K$ is the kernel function for the space $H$. In particular, given $g \in H$,
\begin{gather*}
\langle A_f g, \Gamma_{\gamma_j} \rangle_{\tilde H} = \int_0^{T_j} \nabla g(\gamma_j(t)) f(\gamma_j(t)) dt\\ = \int_0^{T_j} \dot g(\gamma_j(t)) dt = g(\gamma_j(T_j)) - g(\gamma_j(0)) = \langle g, K_{\gamma_j(T)} - K_{\gamma_j(0)} \rangle_{H}.
\end{gather*}

The objective is to construct a finite rank approximation of $A_f$ through which an SVD may be performed to find approximate singular values and singular vectors, and to ultimately approximate the singular Liouville modes. Note that since the dynamics are unknown, the adjoint must be approximated instead, where the action of the adjoint on the occupation kernels provides a sample of the operator. Thus, the finite rank representation will be determined through the restriction of $H$ to the span of the ordered basis $\alpha = \{ \Gamma_{\gamma_j} \}_{j=1}^M$. A corresponding basis for $H$ must also be selected, and given the available information, $\beta = \{ K(\cdot,\gamma_j(T_j)) - K(\cdot,\gamma_j(0))\}$ is most reasonable. Of course, this leads to a rather benign matrix representation of
\[ [A_f^*]_{\alpha}^\beta = \begin{pmatrix} 1 & & \\ & \ddots & \\ & & 1\end{pmatrix}.\]
Moreover, if this matrix is input into an SVD routine, typical algorithms would not be aware of the non-orthogonal inner products between the basis elements. To rectify this, two orthonormal bases $\alpha'$ and $\beta'$ may be obtained from an eigendecomposition of the Gram matrices (which are assumed to be strictly positive definite) for $\alpha$ and $\beta$ respectively:
\begin{gather*}
\tilde G := \begin{pmatrix} \langle \Gamma_{\gamma_1}, \Gamma_{\gamma_1} \rangle_{\tilde H} & \cdots & \langle \Gamma_{\gamma_1}, \Gamma_{\gamma_M} \rangle_{\tilde H}\\
\vdots & \ddots & \vdots\\
\langle \Gamma_{\gamma_M}, \Gamma_{\gamma_1} \rangle_{\tilde H} & \cdots & \langle \Gamma_{\gamma_M}, \Gamma_{\gamma_M} \rangle_{\tilde H}
\end{pmatrix} = V \Lambda V^*:=\\
\begin{pmatrix} | & & |\\
\tilde v_1 & \cdots & \tilde v_M\\
| & & |
\end{pmatrix}
\begin{pmatrix} \tilde \lambda_1 & & \\
 & \cdots & \\
 & & \tilde \lambda_M
\end{pmatrix}
\begin{pmatrix} - &\tilde v_1^* & -\\
 & \vdots & \\
 - & \tilde v_M^* & -
\end{pmatrix},\text{ and}
\\
G = \begin{pmatrix} \langle A_f^* \Gamma_{\gamma_1}, A_f^* \Gamma_{\gamma_1} \rangle_{ H} & \cdots & \langle A_f^* \Gamma_{\gamma_1}, A_f^* \Gamma_{\gamma_M} \rangle_{ H}\\
\vdots & \ddots & \vdots\\
\langle A_f^* \Gamma_{\gamma_M},A_f^*  \Gamma_{\gamma_1} \rangle_{ H} & \cdots & \langle A_f^* \Gamma_{\gamma_M}, A_f^* \Gamma_{\gamma_M} \rangle_{ H}\\
\end{pmatrix} = \tilde V \tilde \Lambda \tilde V^*:=\\\\ =
\begin{pmatrix} | & & |\\
 v_1 & \cdots &  v_M\\
| & & |
\end{pmatrix}
\begin{pmatrix}  \lambda_1 & & \\
 & \cdots & \\
 & &  \lambda_M
\end{pmatrix}
\begin{pmatrix} - & v_1^* & -\\
 & \vdots & \\
 - &  v_M^* & -
\end{pmatrix}.
\end{gather*}
A more meaningful representation of $A_f^*$ may be found by re-expressing $[A_f^*]_{\alpha}^{\beta}$ in terms of the orthornormal sets $\alpha' = \{ q_j \}_{j=1}^M$ and $\beta' = \{ p_j \}_{j=1}^M$ where
\begin{gather*} p_j = \frac{1}{\sqrt{v_j^* G v_j}} \sum_{\ell = 1}^M (v_j)_{\ell} ( K(\cdot,\gamma_\ell(T_\ell)) - K(\cdot,\gamma_\ell(0))), \text{ and}\\
q_j = \frac{1}{\sqrt{\tilde v_j^* \tilde G \tilde v_j}} \sum_{\ell=1}^M (\tilde v_j)_{\ell} \Gamma_{\gamma_\ell}.\end{gather*}
In other words,
\[
\begin{pmatrix}
q_1(x) \\ \vdots \\ q_M(x)
\end{pmatrix}
=
\begin{pmatrix}
\left(\sqrt{\tilde v_1^* \tilde G \tilde v_1}\right)^{-1} & &\\
& \ddots &\\
& & \left(\sqrt{\tilde v_M^* \tilde G \tilde v_M}\right)^{-1}
\end{pmatrix}
\tilde V^T
\begin{pmatrix}
\Gamma_{\gamma_1}(x)\\
\vdots\\
\Gamma_{\gamma_M}(x)
\end{pmatrix},
\]
and a similar expression may be written for $p_j$. Write
\begin{gather*}\tilde V_0 = \tilde V \diag\left(\left(\sqrt{\tilde v_1^* \tilde G \tilde v_1}\right)^{-1}, \ldots, \left(\sqrt{\tilde v_M^* \tilde G \tilde v_M}\right)^{-1}\right),\text{ and}\\
V_0 = V \diag\left(\left(\sqrt{ v_1^*  G  v_1}\right)^{-1}, \ldots, \left(\sqrt{ v_M^*  G  v_M}\right)^{-1}\right)\end{gather*}
the coefficients of each column of $V_0$ and $\tilde V_0$ correspond to functions of norm 1 for their respective spaces. It follows that
\begin{gather*}
    [A_f^*]_{\alpha'}^{\beta'} = V_0^{-1} [A_f^*]_{\alpha}^\beta \tilde V_0 = V_0^{-1} \tilde V_0.
\end{gather*}
That is, the matrix representation with respect to the bases $\beta'$ and $\alpha'$ are obtained by sending elements of $\alpha'$ to $\alpha$, computing the action of $[A_f^*]_{\alpha}^\beta$ on this transformation, and then sending the result expressed in terms of the $\beta$ basis to $\beta'$.

Now the approximate singular vectors may be obtained for $A_f$ by taking the SVD of $[A_f^*]_{\alpha'}^{\beta'}$. In particular, the right singular vectors of $[A_f^*]_{\alpha'}^{\beta'}$ will be correspond to the approximate left singular functions of $A_f$ and vice versa. In particular, writing the SVD of $[A_f^*]_{\alpha'}^{\beta'}$ as
\begin{gather*}
    [A_f^*]_{\alpha'}^{\beta'} = \hat U \hat \Sigma \hat V^* =
    \begin{pmatrix} | & & |\\
        \hat u_1 & \cdots & \hat u_M\\
        | & & |
    \end{pmatrix}
    \begin{pmatrix}
        \hat \sigma^2_1 & & \\
        & \ddots &\\
        & & \hat \sigma^2_M
    \end{pmatrix}
    \begin{pmatrix}
        - & \hat v_1^* & -\\
        & \vdots & \\
        - & \hat v_M^* & -
    \end{pmatrix},
\end{gather*}
and the approximate right singular vector for $A_f$ is $\hat \varphi_j = \frac{1}{\sqrt{u_j^* G_p u_j}}\sum_{\ell} (\hat u_j )_\ell p_\ell$, and the approximate left singular vector for $A_f$ is $\hat \psi_j = \frac{1}{\sqrt{\hat v_j^* G_q \hat v_j}}\sum_{\ell} (\hat v_j )_\ell q_\ell,$ where $G_p$ and $G_q$ are the Gram matrices for the ordered bases $\beta'$ and $\alpha'$ respectively.

Translating this to the original bases $\alpha$ and $\beta$, we find the following:
\begin{gather*}
    \hat \varphi_j = \frac{1}{\sqrt{\hat u_j^* G_p \hat u_j}} u_j^T V_0^T \begin{pmatrix} K(\cdot, \gamma_1(T_1)) - K(\cdot,\gamma_1(0))\\ \vdots \\ K(\cdot,\gamma_M(T_M)) - K(\cdot,\gamma_M(0))
    \end{pmatrix},\text{ and}\\
    \hat \psi_j = \frac{1}{\sqrt{\hat v_j^* G_q \hat v_j}} \hat v_j^T \tilde V_0^T \begin{pmatrix} \Gamma_{\gamma_1}\\ \vdots \\ \Gamma_{\gamma_M}
    \end{pmatrix}.
\end{gather*}

Thus, if $x : [0,T] \to \mathbb{R}^n$ satisfies $\dot x = f(x)$, then it may be approximately expressed through the following integral equation:
\begin{gather*}
    x(t) \approx x(0) + \int_0^t \sum_{j=1}^M \hat \xi_j \hat \psi_j(x(\tau)) d\tau,
\end{gather*}
where
\begin{gather*}\hat \xi_j = \begin{pmatrix} \langle (x)_1,\hat \varphi_j \rangle_H\\ \vdots\\ \langle (x)_n, \hat \varphi_j \rangle_H\end{pmatrix}, \text{ and}\\
\langle (x)_i, \hat \phi_j \rangle_H = \diag\left(\frac{1}{\sqrt{\hat u_1^* G_p \hat u_1}}, \cdots, \frac{1}{\sqrt{\hat u_M^* G_p \hat u_M}}\right)\\ \times
\begin{pmatrix} - & \hat u_1^T & -\\
& \vdots & \\
 - & \hat u_M^T & -\end{pmatrix}
V_0^T
\begin{pmatrix} (\gamma_1(T_1))_i - (\gamma_1(0))_i\\ \vdots \\ (\gamma_M(T_M))_i - (\gamma_M(0))_i
    \end{pmatrix}.
    \end{gather*}

\section{The Eigenfunction based DMD Algorithm}\label{sec:eigenfunctiondmd}
In this section it will be assumed that $A_f : H \to \tilde H$ is a compact operator, and that $H \subset \tilde H$. Since $A_f$ is compact, it is bounded, which means that unlike \cite{gonzalez2021anti} and \cite{rosenfeld2019dynamic}, no additional assumptions are needed concerning the domain of this operator.

For a collection of observed trajectories, $\{ \gamma_1,\ldots,\gamma_M\}$ consider the collection of occupation kernels, $\alpha = \{\Gamma_{\gamma_1}, \ldots, \Gamma_{\gamma_M} \}_{m=1}^M$, where these are the occupation kernels for the space $H$, and let $\beta = \{ \tilde \Gamma_{\gamma_1}, \ldots, \tilde \Gamma_{\gamma_M}\}$ be the occupation kernels in $\tilde H$. Let $P_\alpha$ be the projection from $H$ to $H$ onto the span of $\alpha$, and let $\tilde P_{\alpha}$ and $\tilde P_{\beta}$ be the corresponding projections onto the spans of $\alpha$ and $\beta$ respectively (viewed as subspaces of $\tilde H$). The numerical method presented in this section will construct a matrix representation for the operator $\tilde P_{\alpha} \tilde P_{\beta} A_f P_{\alpha}$, where the matrix, $[\tilde P_{\alpha}\tilde P_{\beta}  A_f P_{\alpha}]_{\alpha}^\alpha$, represents this operator on the span of $\alpha$ in the domain and range respectively. Note that since the matrix representation is defined over $\alpha$, $[\tilde P_{\alpha} \tilde P_{\beta} A_f P_{\alpha}]_{\alpha}^\alpha = [\tilde P_{\alpha}\tilde P_{\beta}  A_f]_{\alpha}^\alpha.$

Recall that for a function $g \in \tilde H$, $\tilde P_{\beta} g$ is a linear combination of the functions of $\alpha$ as $\sum_{m=1}^M w_m \Gamma_{\gamma_m}$, where the weights are obtained via
\[
\begin{pmatrix}
\langle \tilde \Gamma_{\gamma_1},\tilde \Gamma_{\gamma_1} \rangle_{\tilde H} & \cdots & \langle \tilde \Gamma_{\gamma_1},\tilde \Gamma_{\gamma_M} \rangle_{\tilde H}\\
\vdots & \ddots & \vdots\\
\langle \tilde \Gamma_{\gamma_M},\tilde \Gamma_{\gamma_1} \rangle_{\tilde H} & \cdots & \langle \tilde \Gamma_{\gamma_M},\tilde \Gamma_{\gamma_M} \rangle_{\tilde H}
\end{pmatrix}
\begin{pmatrix}w_1 \\ \vdots \\ w_M \end{pmatrix}
=
\begin{pmatrix}\langle g, \tilde \Gamma_{\gamma_1} \rangle_{\tilde H} \\ \vdots \\ \langle g, \tilde \Gamma_{\gamma_M} \rangle_{\tilde H} \end{pmatrix},
\]
and the matrix is called the Gram matrix for the basis $\alpha$ in the space $\tilde H$.

Hence, for each $\Gamma_{\gamma_j}$, the weights for the projection of $A_f \Gamma_{\gamma_j}$ onto the span of $\beta$ may be obtained as
\begin{gather*}
\begin{pmatrix}
\langle \tilde \Gamma_{\gamma_1},\tilde \Gamma_{\gamma_1} \rangle_{\tilde H} & \cdots & \langle \tilde \Gamma_{\gamma_1},\tilde \Gamma_{\gamma_M} \rangle_{\tilde H}\\
\vdots & \ddots & \vdots\\
\langle \tilde \Gamma_{\gamma_M},\tilde \Gamma_{\gamma_1} \rangle_{\tilde H} & \cdots & \langle \tilde \Gamma_{\gamma_M},\tilde \Gamma_{\gamma_M} \rangle_{\tilde H}
\end{pmatrix}
\begin{pmatrix}w_1 \\ \vdots \\ w_M \end{pmatrix}\\
=
\begin{pmatrix}\langle A_f \Gamma_{\gamma_j}, \tilde \Gamma_{\gamma_1} \rangle_{\tilde H} \\ \vdots \\ \langle A_f \Gamma_{\gamma_j}, \tilde \Gamma_{\gamma_M} \rangle_{\tilde H} \end{pmatrix}
=
\begin{pmatrix}\langle \Gamma_{\gamma_j}, A_f^* \tilde \Gamma_{\gamma_1} \rangle_{H} \\ \vdots \\ \langle \Gamma_{\gamma_j}, A_f^* \tilde \Gamma_{\gamma_M} \rangle_{H} \end{pmatrix}\\
=
\begin{pmatrix}\langle \Gamma_{\gamma_j}, K_{\gamma_1(T_1)} - K_{\gamma_1(0)} \rangle_{H} \\ \vdots \\ \langle \Gamma_{\gamma_j},  K_{\gamma_M(T_M)} - K_{\gamma_M(0)} \rangle_{H} \end{pmatrix}
= 
\begin{pmatrix}\Gamma_{\gamma_j}(\gamma_1(T_1)) - \Gamma_{\gamma_j}(\gamma_1(0)) \\ \vdots \\ \Gamma_{\gamma_j}(\gamma_M(T_M)) - \Gamma_{\gamma_j}(\gamma_M(0)) \end{pmatrix}.
\end{gather*}
Next, a projection onto the span of $\alpha$ within $\tilde H$ must be performed. For each $\tilde \Gamma_{\gamma_j}$, the weights corresponding to its projection onto $\alpha$ are given via
\begin{gather*}
\begin{pmatrix}
\langle \Gamma_{\gamma_1}, \Gamma_{\gamma_1} \rangle_{\tilde H} & \cdots & \langle \Gamma_{\gamma_1}, \Gamma_{\gamma_M} \rangle_{\tilde H}\\
\vdots & \ddots & \vdots\\
\langle \Gamma_{\gamma_M}, \Gamma_{\gamma_1} \rangle_{\tilde H} & \cdots & \langle \Gamma_{\gamma_M}, \Gamma_{\gamma_M} \rangle_{\tilde H}
\end{pmatrix}
\begin{pmatrix}v_{1,j} \\ \vdots \\ v_{M,j}\end{pmatrix}
=
\begin{pmatrix}\langle \tilde \Gamma_{\gamma_j}, \Gamma_{\gamma_1} \rangle_{\tilde H} \\ \vdots \\ \langle \tilde  \Gamma_{\gamma_j}, \Gamma_{\gamma_M} \rangle_{\tilde H} \end{pmatrix}
\end{gather*}

Hence, the projection of $A_{f} \Gamma_{\gamma_j}$ is given as
\begin{gather*}
\tilde P_{\alpha} \tilde P_{\beta} A_f \Gamma_{\gamma_j} = \sum_{m=1}^M w_{m} \sum_{\ell=1}^M v_{\ell,m} \Gamma_{\gamma_{\ell}} =\\
\sum_{m=1}^M w_{m} \sum_{\ell=1}^M\left(
\begin{pmatrix}
\langle \Gamma_{\gamma_1}, \Gamma_{\gamma_1} \rangle_{\tilde H} & \cdots & \langle \Gamma_{\gamma_1}, \Gamma_{\gamma_M} \rangle_{\tilde H}\\
\vdots & \ddots & \vdots\\
\langle \Gamma_{\gamma_M}, \Gamma_{\gamma_1} \rangle_{\tilde H} & \cdots & \langle \Gamma_{\gamma_M}, \Gamma_{\gamma_M} \rangle_{\tilde H}
\end{pmatrix}^{-1}
\begin{pmatrix}\langle \tilde \Gamma_{\gamma_m}, \Gamma_{\gamma_1} \rangle_{\tilde H} \\ \vdots \\ \langle \tilde  \Gamma_{\gamma_m}, \Gamma_{\gamma_M} \rangle_{\tilde H} \end{pmatrix}
\right)^T
\begin{pmatrix}\Gamma_{\gamma_1} \\ \vdots \\ \Gamma_{\gamma_M} \end{pmatrix}\\
=
\left(
\begin{pmatrix}
\langle \tilde \Gamma_{\gamma_1},\tilde \Gamma_{\gamma_1} \rangle_{\tilde H} & \cdots & \langle \tilde \Gamma_{\gamma_1},\tilde \Gamma_{\gamma_M} \rangle_{\tilde H}\\
\vdots & \ddots & \vdots\\
\langle \tilde \Gamma_{\gamma_M},\tilde \Gamma_{\gamma_1} \rangle_{\tilde H} & \cdots & \langle \tilde \Gamma_{\gamma_M},\tilde \Gamma_{\gamma_M} \rangle_{\tilde H}
\end{pmatrix}^{-1}
\begin{pmatrix}\Gamma_{\gamma_j}(\gamma_1(T_1)) - \Gamma_{\gamma_j}(\gamma_1(0)) \\ \vdots \\ \Gamma_{\gamma_j}(\gamma_M(T_M)) - \Gamma_{\gamma_j}(\gamma_M(0)) \end{pmatrix}
\right)^T
\\
\times
\left(
\begin{pmatrix}
\langle \Gamma_{\gamma_1}, \Gamma_{\gamma_1} \rangle_{\tilde H} & \cdots & \langle \Gamma_{\gamma_1}, \Gamma_{\gamma_M} \rangle_{\tilde H}\\
\vdots & \ddots & \vdots\\
\langle \Gamma_{\gamma_M}, \Gamma_{\gamma_1} \rangle_{\tilde H} & \cdots & \langle \Gamma_{\gamma_M}, \Gamma_{\gamma_M} \rangle_{\tilde H}
\end{pmatrix}^{-1}
\begin{pmatrix}\langle \tilde \Gamma_{\gamma_1}, \Gamma_{\gamma_1} \rangle_{\tilde H} & \cdots & \langle \tilde \Gamma_{\gamma_M}, \Gamma_{\gamma_1} \rangle_{\tilde H}\\
\vdots \\
\langle \tilde \Gamma_{\gamma_1}, \Gamma_{\gamma_M} \rangle_{\tilde H} & \cdots & \langle \tilde \Gamma_{\gamma_M}, \Gamma_{\gamma_M} \rangle_{\tilde H}\end{pmatrix}
\right)^T
\begin{pmatrix}\Gamma_{\gamma_1} \\ \vdots \\ \Gamma_{\gamma_M} \end{pmatrix},
\end{gather*}
and the final representation, $[\tilde P_{\alpha} \tilde P_{\beta} A_f ]_{\alpha}^\alpha$ is given as
\begin{gather}\label{eq:finiterankrep}
    [\tilde P_{\alpha} \tilde P_{\beta} A_f ]_{\alpha}^\alpha =
\\ \nonumber
    \begin{pmatrix}
        \langle \Gamma_{\gamma_1}, \Gamma_{\gamma_1} \rangle_{\tilde H} & \cdots & \langle \Gamma_{\gamma_1}, \Gamma_{\gamma_M} \rangle_{\tilde H}\\
        \vdots & \ddots & \vdots\\
        \langle \Gamma_{\gamma_M}, \Gamma_{\gamma_1} \rangle_{\tilde H} & \cdots & \langle \Gamma_{\gamma_M}, \Gamma_{\gamma_M} \rangle_{\tilde H}
    \end{pmatrix}^{-1}
    \begin{pmatrix}
        \langle \tilde \Gamma_{\gamma_1}, \Gamma_{\gamma_1} \rangle_{\tilde H} & \cdots & \langle \tilde \Gamma_{\gamma_M}, \Gamma_{\gamma_1} \rangle_{\tilde H}\\
        \vdots \\
        \langle \tilde \Gamma_{\gamma_1}, \Gamma_{\gamma_M} \rangle_{\tilde H} & \cdots & \langle \tilde \Gamma_{\gamma_M}, \Gamma_{\gamma_M} \rangle_{\tilde H}
    \end{pmatrix}
\\  \nonumber
    \times
    \begin{pmatrix}
        \langle \tilde \Gamma_{\gamma_1},\tilde \Gamma_{\gamma_1} \rangle_{\tilde H} & \cdots & \langle \tilde \Gamma_{\gamma_1},\tilde \Gamma_{\gamma_M} \rangle_{\tilde H}\\
        \vdots & \ddots & \vdots\\
        \langle \tilde \Gamma_{\gamma_M},\tilde \Gamma_{\gamma_1} \rangle_{\tilde H} & \cdots & \langle \tilde \Gamma_{\gamma_M},\tilde \Gamma_{\gamma_M} \rangle_{\tilde H}
    \end{pmatrix}^{-1}
\\  \nonumber
    \times
    \begin{pmatrix}
        \Gamma_{\gamma_1}(\gamma_1(T_1)) - \Gamma_{\gamma_1}(\gamma_1(0)) & \cdots & \Gamma_{\gamma_M}(\gamma_1(T_1)) - \Gamma_{\gamma_M}(\gamma_1(0))
        \\ \vdots \\ 
        \Gamma_{\gamma_1}(\gamma_M(T_M)) - \Gamma_{\gamma_1}(\gamma_M(0)) & \cdots & \Gamma_{\gamma_M}(\gamma_M(T_M)) - \Gamma_{\gamma_M}(\gamma_M(0))
    \end{pmatrix}.
\end{gather}
Note that when $H = \tilde H$ and the occupation kernels are assumed to be in the domain of the Liouville operator, the first two matrices cancel, and the representation reduces to that of \cite{rosenfeld2019dynamic}.

Under the assumption of diagonalizability for \eqref{eq:finiterankrep}, which holds for almost all matrices, an eigendecomposition for \eqref{eq:finiterankrep} may be determined as
\[
[\tilde P_{\alpha} \tilde P_{\beta} A_f ]_{\alpha}^\alpha =
\begin{pmatrix}
| & & |\\
V_1 & \cdots & V_M\\
| & & |
\end{pmatrix}
\begin{pmatrix}
\lambda_1 & & \\
& \ddots & \\
& & \lambda_M
\end{pmatrix}
\begin{pmatrix}
| & & |\\
V_1 & \cdots & V_M\\
| & & |
\end{pmatrix}^{-1},
\]
where each column, $V_j$, is an eigenvector of $[\tilde P_{\alpha} \tilde P_{\beta} A_f ]_{\alpha}^\alpha$ with eigenvalue $\lambda_j$. The corresponding normalized eigenfunction is given as
\[ \hat \varphi_j(x) = \frac{1}{\sqrt{V_j^T G_{\alpha} V_j}} V_j^T \begin{pmatrix} \Gamma_{\gamma_1} \\ \vdots \\ \Gamma_{\gamma_M}\end{pmatrix},\]
where the normalization is performed in the Hilbert space $H$ through the Gram matrix for $\alpha$, $G_\alpha$, according to $H$'s inner product. Set $\bar V_j := \frac{1}{\sqrt{V_j^T G_{\alpha} V_j}} V_j$, and let $\bar V := \left( V_1 \cdots V_M \right)$.

The Gram matrix for the normalized eigenbasis may be quickly computed as $\bar V^T G_{\alpha} \bar V$, and the weights for the projection of the full state observable onto this eigenbasis may be written as
\begin{gather*}
\begin{pmatrix}
- & \hat \xi_1^T & -\\
& \vdots & \\
- & \hat \xi_M^T & -
\end{pmatrix}
=
(\bar V^T G_{\alpha} \bar V)^{-1}
\begin{pmatrix}
    \langle (x)_1, \hat \varphi_1 \rangle_H & \cdots & \langle(x)_n, \hat \varphi_1 \rangle_H\\
    \vdots & \ddots & \vdots\\
    \langle (x)_1, \hat \varphi_M \rangle_H & \cdots & \langle(x)_n, \hat \varphi_M \rangle_H
\end{pmatrix}
\\
=
(\bar V^T G_{\alpha} \bar V)^{-1} \bar V^T
\begin{pmatrix}
    \langle (x)_1, \Gamma_{\gamma_1} \rangle_H & \cdots & \langle(x)_n, \Gamma_{\gamma_1} \rangle_H\\
    \vdots & \ddots & \vdots\\
    \langle (x)_1, \Gamma_{\gamma_M} \rangle_H & \cdots & \langle(x)_n, \Gamma_{\gamma_M} \rangle_H
\end{pmatrix}
\\
=
(\bar V^T G_{\alpha} \bar V)^{-1} \bar V^T
\begin{pmatrix}
    \int_{0}^{T_1} \gamma_1(t)^T dt\\
    \vdots\\
    \int_{0}^{T_1} \gamma_M(t)^T dt
\end{pmatrix}
\end{gather*}
and thus,
\begin{equation}\label{eq:fullstateprojection} g_{id}(x) \approx \sum_{m=1}^M \hat \xi_{m} \hat \varphi_m(x).\end{equation}
The approximation error (with respect to the norm of the RKHS) approaches zero if the number of trajectories increases and the corresponding collection of occupation kernels forms a dense set. Convergence in the norm of the RKHS implies uniform convergence on compact subsets of the domain.

Consequently, a trajectory $x:[0,T] \to \mathbb{R}^n$ satisfying $\dot x = f(x)$ may be approximately expressed as \begin{equation*}
     x(t) = g_{id}(x(t)) \approx \sum_{m=1}^M \hat \xi_{m} e^{\lambda_m t} \hat \varphi_m(x(0)),
\end{equation*}
where the eigenfunctions for the finite rank approximation of $A_f$ play the role of eigenfunctions for the original operator, $A_f$.

Note that for a given $\epsilon > 0$ there is a sufficiently large collection of trajectories and occupation kernels such that $\| \tilde P_{\alpha} \tilde P_{\beta} A_{f} P_{\alpha} - A_{f} \|_H^{\tilde H} < \epsilon.$ Hence, if $\hat\varphi$ is a normalized eigenfunction for the finite rank representation with eigenvalue $\lambda$, then
\[
\| \lambda \hat\varphi - A_{f} \hat\varphi \|_{\tilde H} = \| \tilde P_{\alpha} \tilde P_{\beta} A_{f} P_{\alpha} \hat\varphi - A_{f} \hat\varphi \|_{\tilde H} \le \epsilon \| \hat\varphi\|_H = \epsilon.
\]
Consequently, given a compact subset of $\mathbb{R}^n$ and a given tolerance, $\epsilon_0$, a finite rank approximation may be selected such that for each normalized eigenfunction the relation $\left| \frac{d}{dt} \hat\varphi(x(t)) - \lambda \hat\varphi(x(t)) \right| < \epsilon_0$ for all $x(t)$ in the compact set. Hence, for sufficiently rich information, $\hat \varphi(x(t)) \approx e^{\lambda t} \hat\varphi(x(0)).$

\section{Computational Remarks for the Eigenfunction Method}

In the above computations, some entries for the matrices require a bit more analysis. Namely, this includes the inner products, $\langle \Gamma_{\gamma_i},\Gamma_{\gamma_j} \rangle_{\tilde H}$ and $\langle \Gamma_{\gamma_i},\tilde \Gamma_{\gamma_j} \rangle_{\tilde H}$. All the other quantities have been discussed at length in \cite{rosenfeld2019occupation,rosenfeld2019occupation2,rosenfeld2019dynamic}.

The second quantity simply utilizes the functional definition of the function $\tilde \Gamma_{\gamma_j}$ as a function in $\tilde H$, $\langle \Gamma_{\gamma_i},\tilde \Gamma_{\gamma_j} \rangle_{\tilde H} = \int_0^{T_j} \Gamma_{\gamma_i}(\gamma_j(t))dt = \int_0^{T_j} \int_0^{T_i} K(\gamma_j(t),\gamma_i(\tau)) d\tau dt,$ where $K$ is the kernel function for $H$. Note that this means $\langle \Gamma_{\gamma_i},\tilde \Gamma_{\gamma_j} \rangle_{\tilde H} = \langle \Gamma_{\gamma_i}, \Gamma_{\gamma_j} \rangle_{ H}$. However, the first quantity is more complicated and is context dependent. In particular, $\Gamma_{\gamma_i}$ is not the occupation kernel corresponding to $\tilde H$, so it's functional relationship cannot be exploited in the same manner. On the other hand, $\Gamma_{\gamma_i}(x) = \int_0^{T_i} K(x,\gamma_i(t))$. To compute the inner product in $\tilde H$, a specific selection of spaces must be considered.

In the particular setting where $H = F_{\mu_1}^2(\mathbb{R}^n)$ and $\tilde H = F_{\mu_2}^2(\mathbb{R}^n)$, with $\mu_1 < \mu_2$, it follows that $\Gamma_{\gamma_i}(x) = \int_0^T e^{\mu_1 x^T \gamma_i(t)}dt.$ Moreover, $K(x,\gamma_i(t)) = e^{\mu_1 x^T \gamma_i(t)} = e^{\mu_2 x^T\left(\frac{\mu_1}{\mu_2} \gamma(t)\right)} = \tilde K(x,(\mu_1/\mu_2)\gamma_i(t))$. Hence, $\Gamma_{\gamma_i}(x) = \tilde \Gamma_{(\mu_1/\mu_2)\gamma_i}(x)$, and
\begin{gather*}
\langle \Gamma_{\gamma_i},\Gamma_{\gamma_j} \rangle_{\tilde H} = \langle\tilde\Gamma_{(\mu_1/\mu_2)\gamma_i},\tilde\Gamma_{(\mu_1/\mu_2)\gamma_j} \rangle_{\tilde H}\\= \int_0^{T_i}\int_0^{T_j} \tilde K((\mu_1/\mu_2)\gamma_i(t),(\mu_1/\mu_2)\gamma_j(\tau)) d\tau dt .
\end{gather*}
\section{Numerical Results}\label{sec:numericalresults}

This section presents the results obtained through implementation of Section \ref{sec:eigenfunctiondmd} with the domain viewed as embedded in the range of the operator. The experiments were performed on the benchmark cyllinder flow data set found in \cite{kutz2016dynamic} by setting $mu_1 = 1/1000$ and $\mu_2 = 1/999$ for the exponential dot product kernel. It should be noted that the timesteps for that data set are $h = 0.02$. The total dataset comprises $151$ snapshots, and the trajectories for the system were selected from strings of adjacent snapshots of length $5$ yielding $147$ trajectories. Computations were performed using Simpson's Rule.

Presented in Figure \ref{fig:LiouvilleModes} are a selection of approximate Liouville modes obtained for this operator through the finite rank approximation determined by Section \ref{sec:eigenfunctiondmd}. Examples of the reconstructed and original data are shown in Figure \ref{fig:reconstruction} and Figure \ref{fig:original}.

\begin{figure}
    \label{fig:LiouvilleModes}
    \centering
    \includegraphics[width=0.48\textwidth]{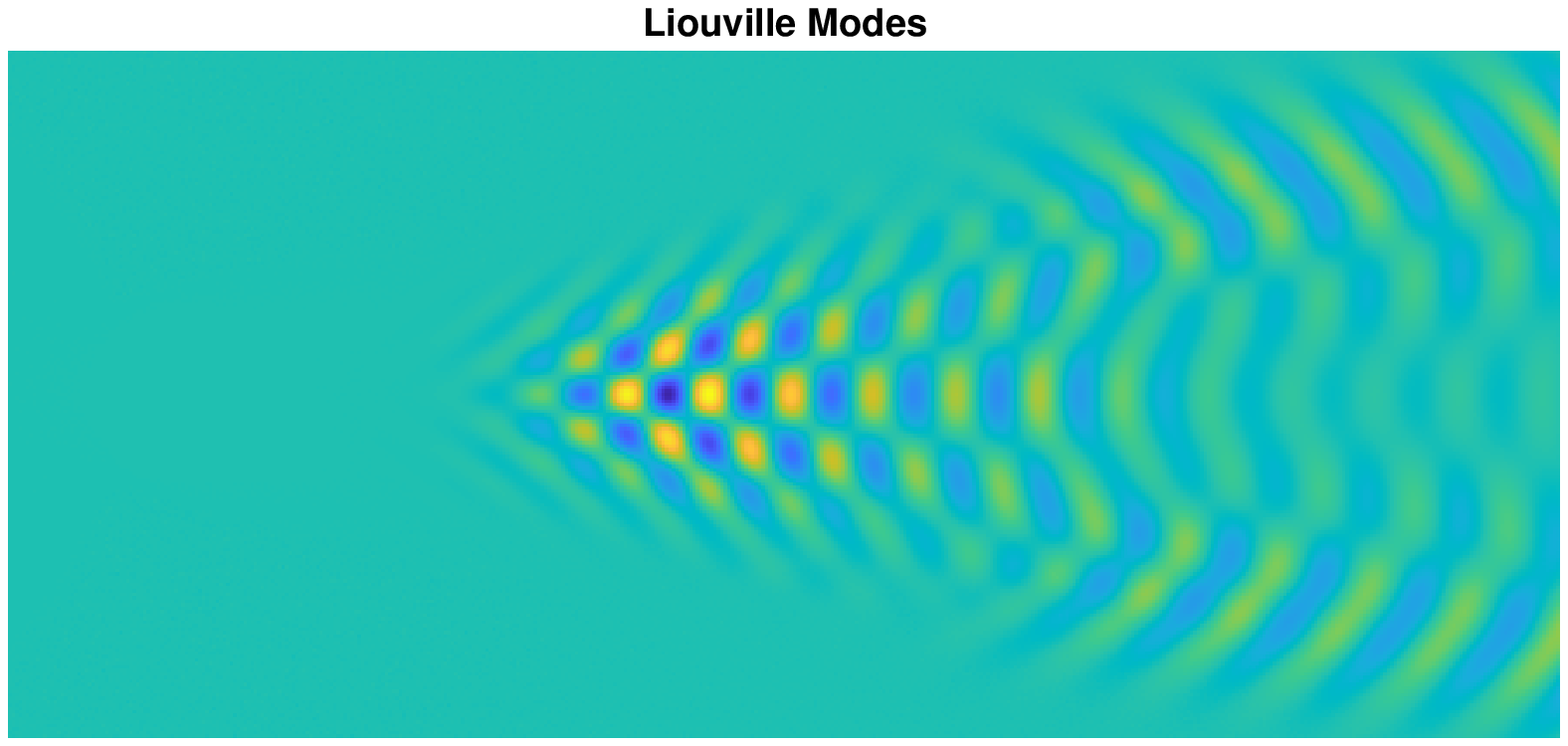}
    \includegraphics[width=0.48\textwidth]{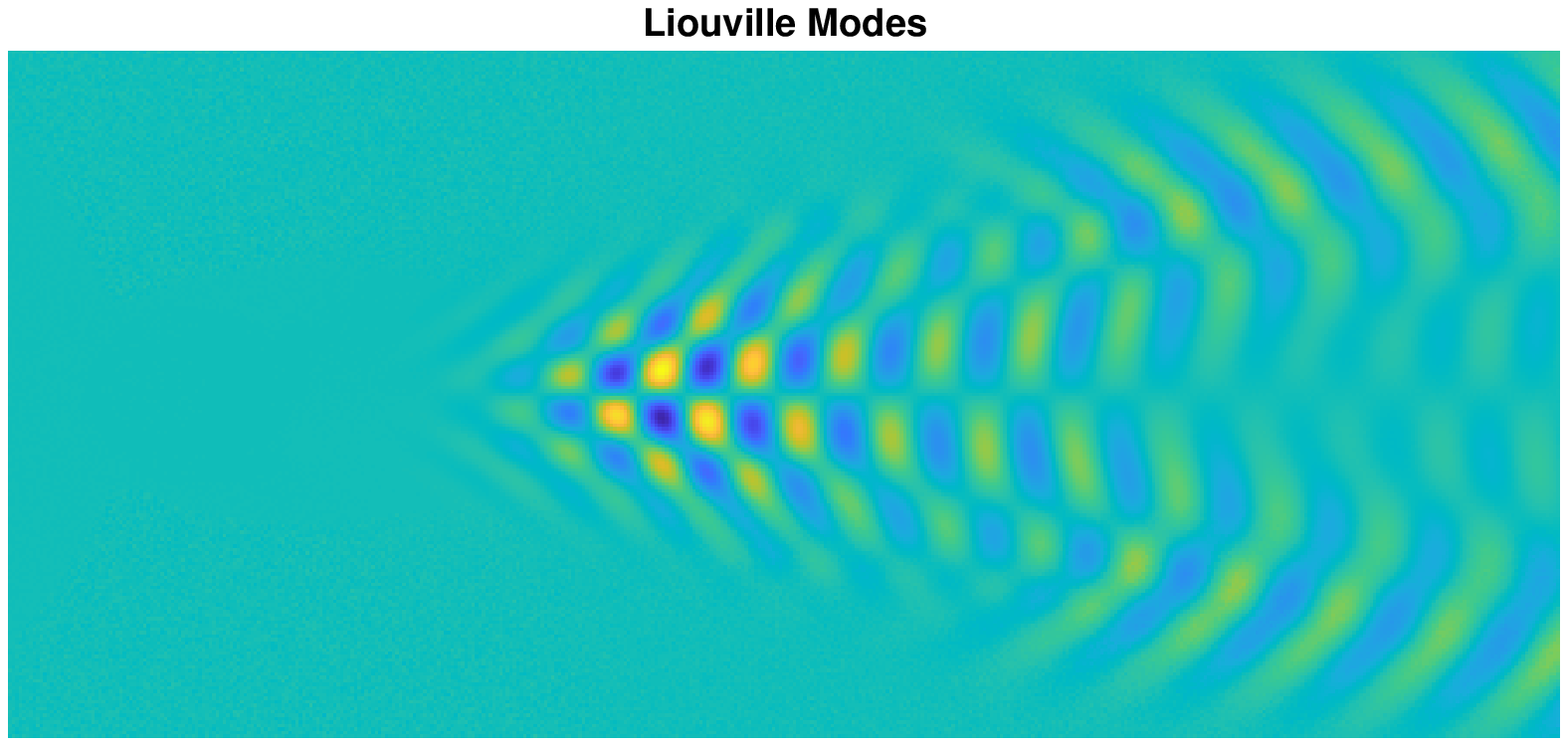}

    \includegraphics[width=0.48\textwidth]{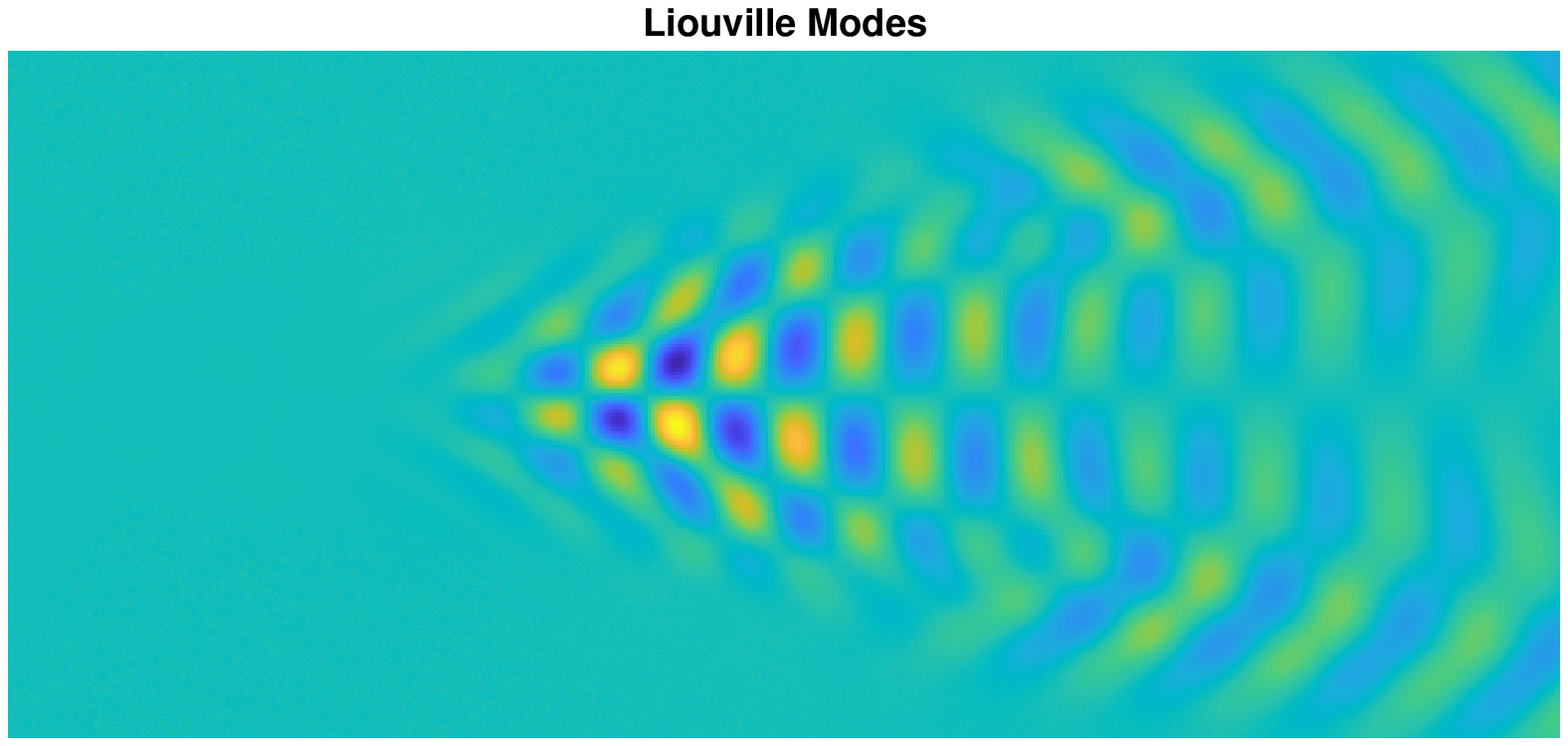}
    \includegraphics[width=0.48\textwidth]{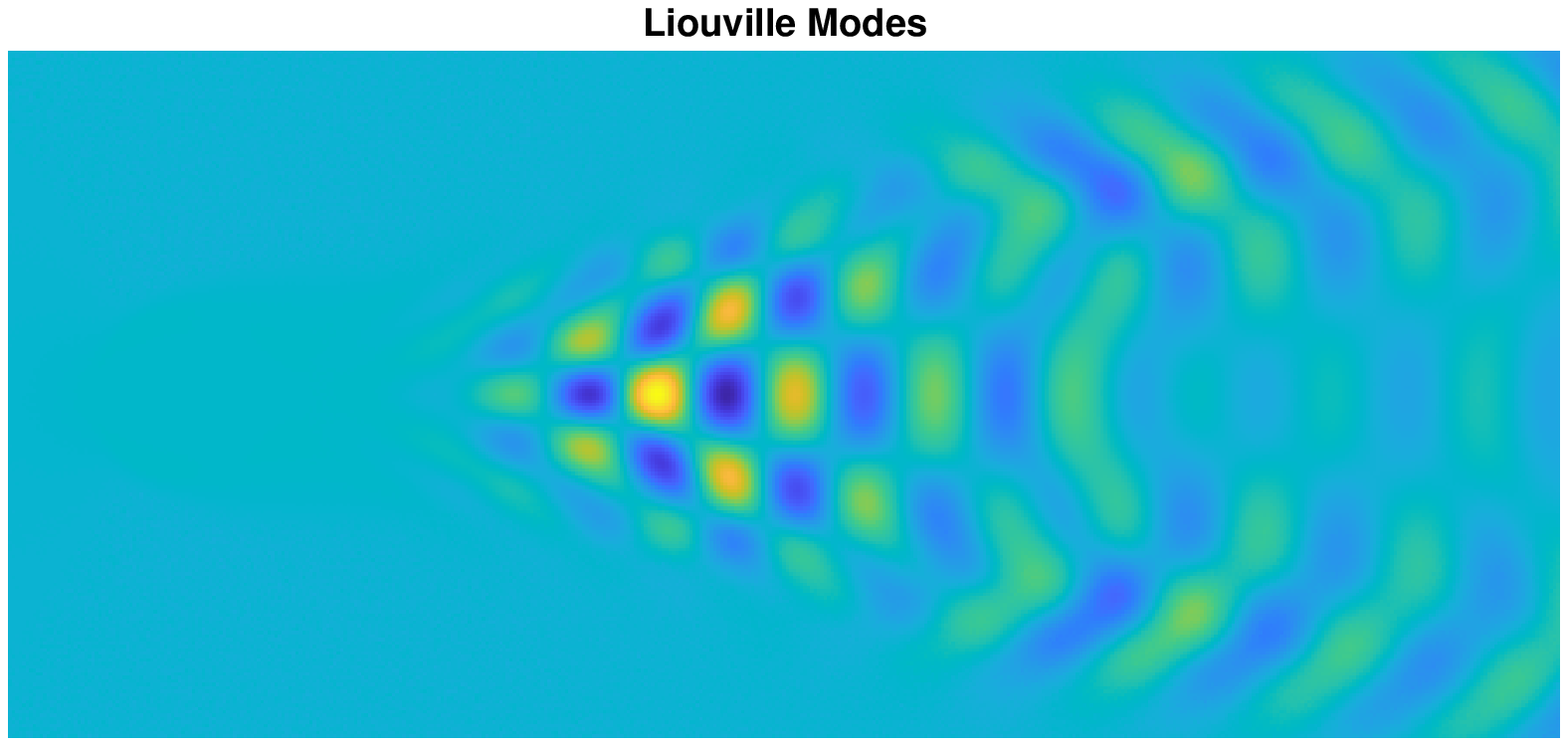}

    \includegraphics[width=0.48\textwidth]{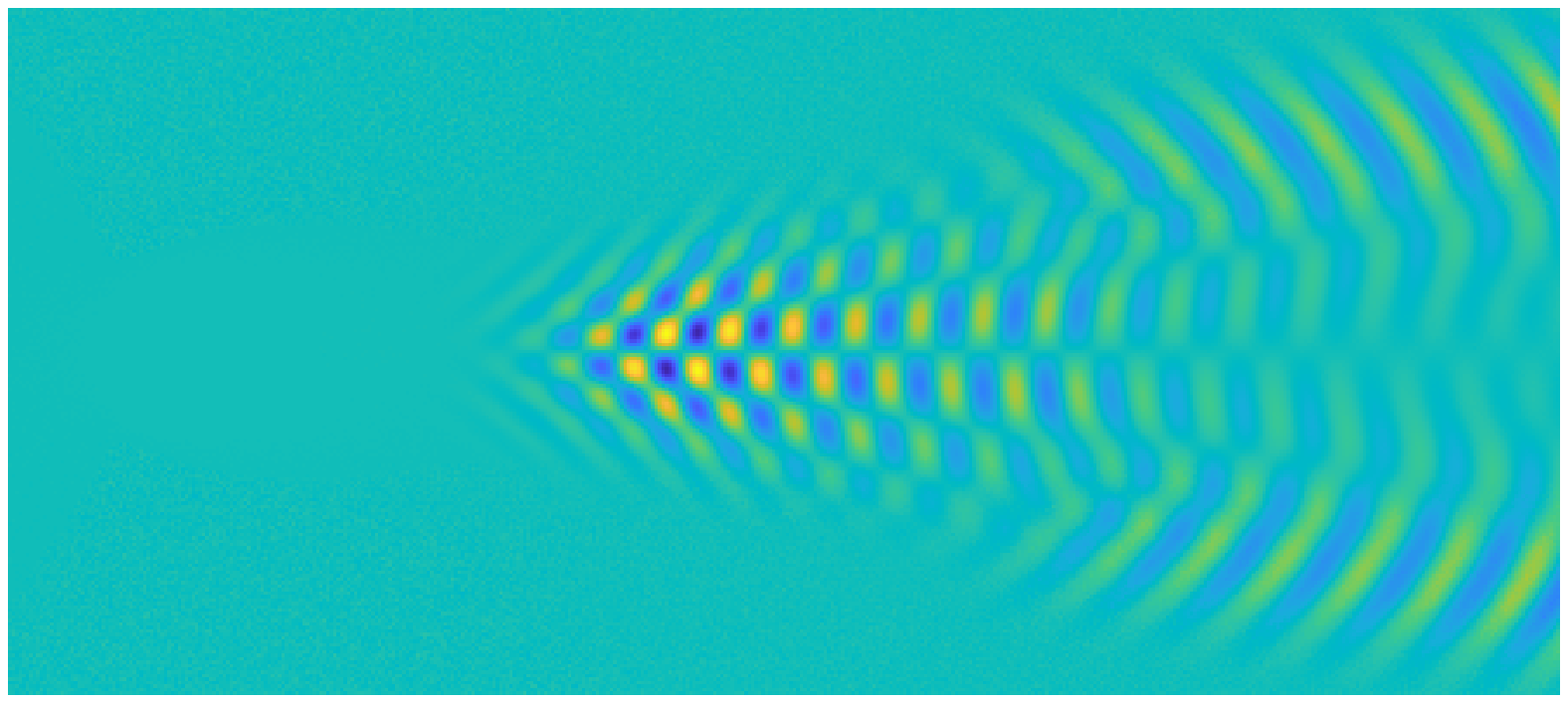}
    \includegraphics[width=0.48\textwidth]{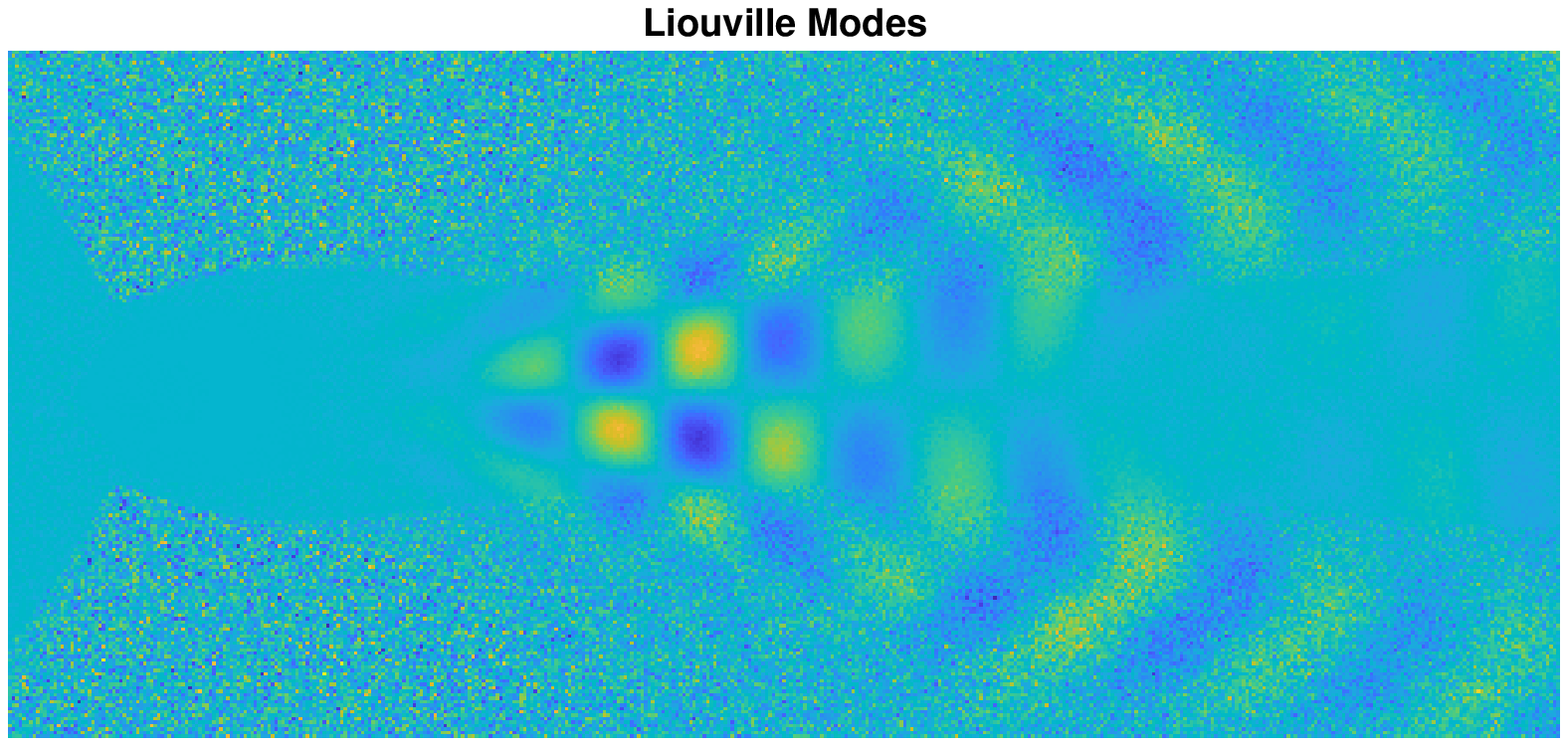}
    
    \includegraphics[width=0.48\textwidth]{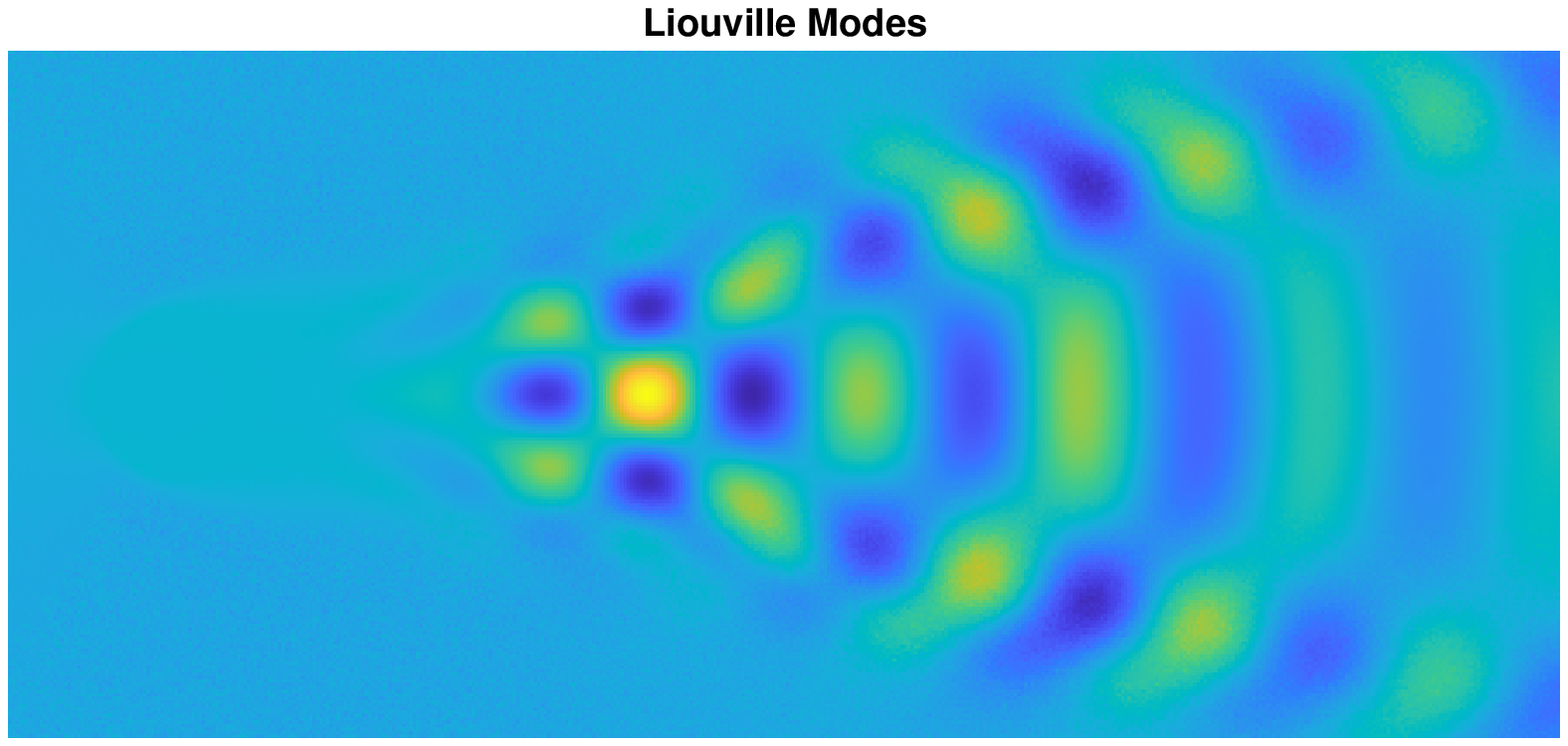}
    \includegraphics[width=0.48\textwidth]{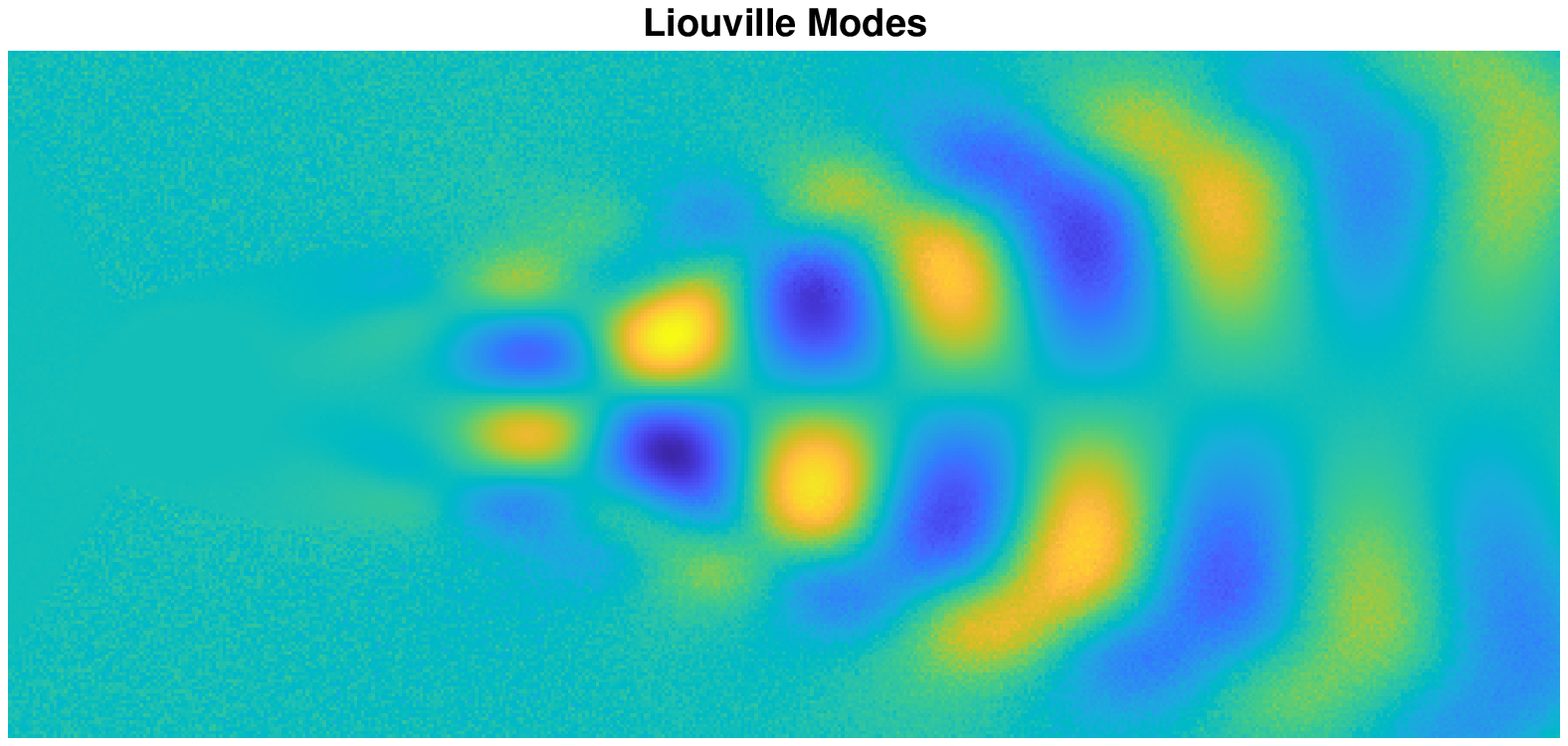}

    \includegraphics[width=0.48\textwidth]{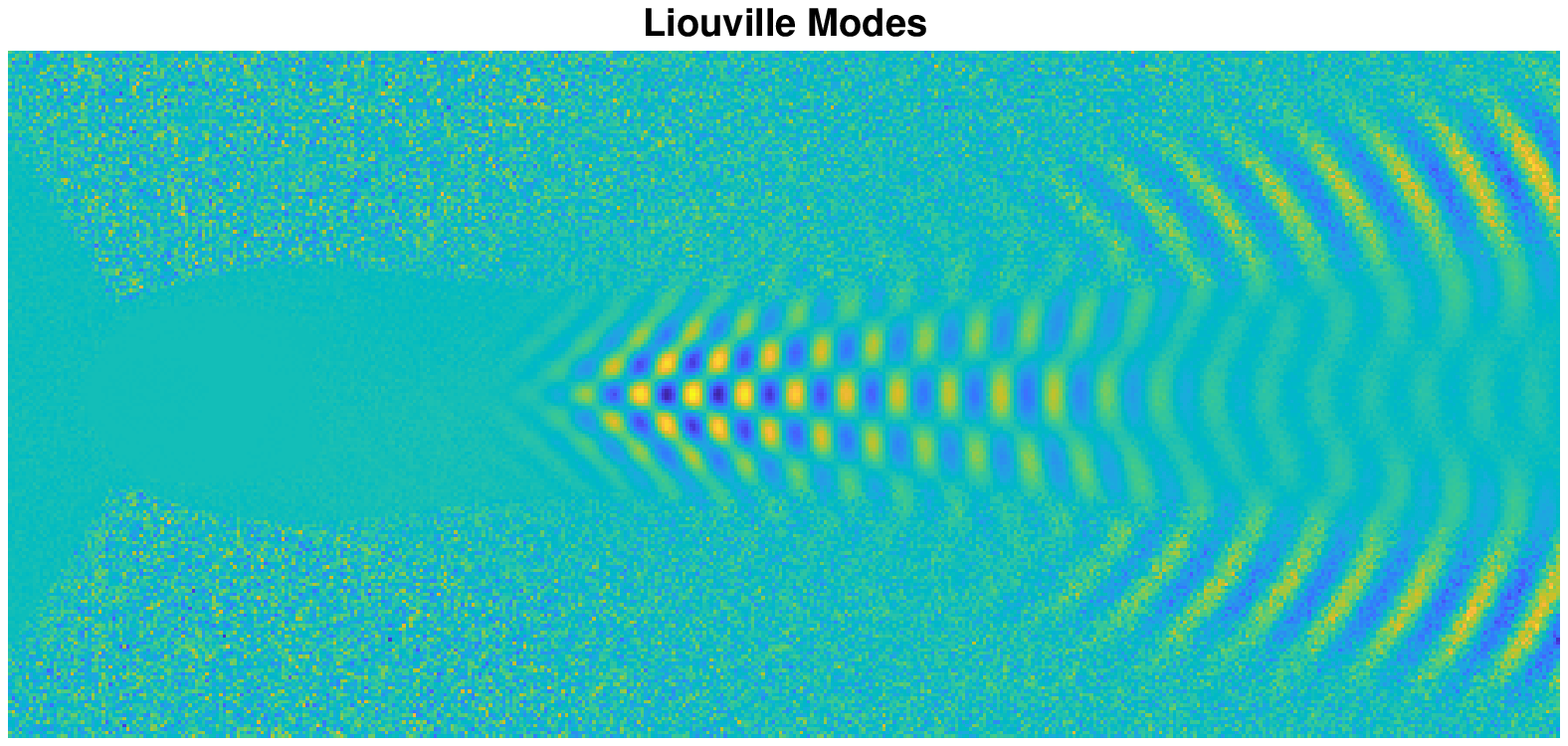}
    \includegraphics[width=0.48\textwidth]{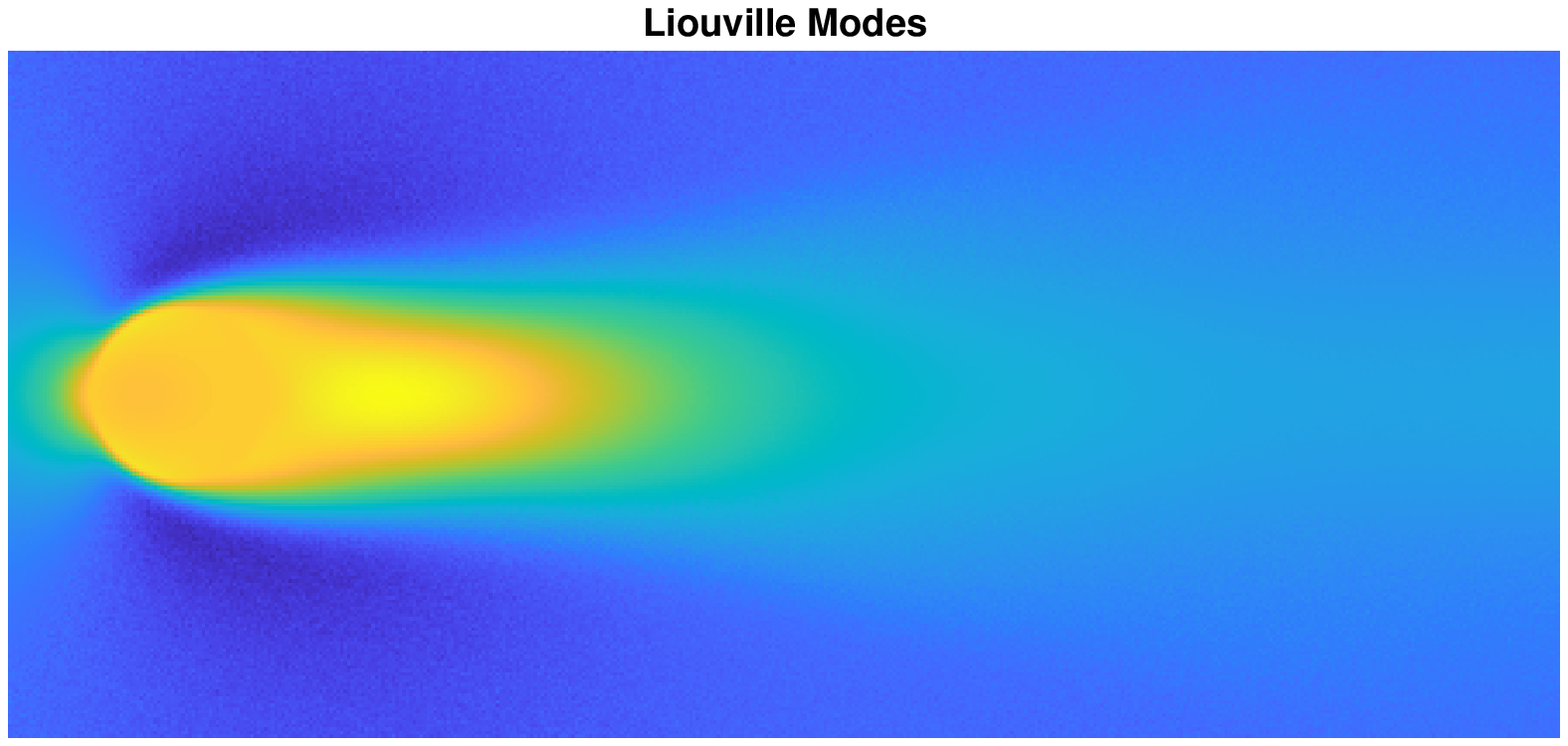}

    \includegraphics[width=0.48\textwidth]{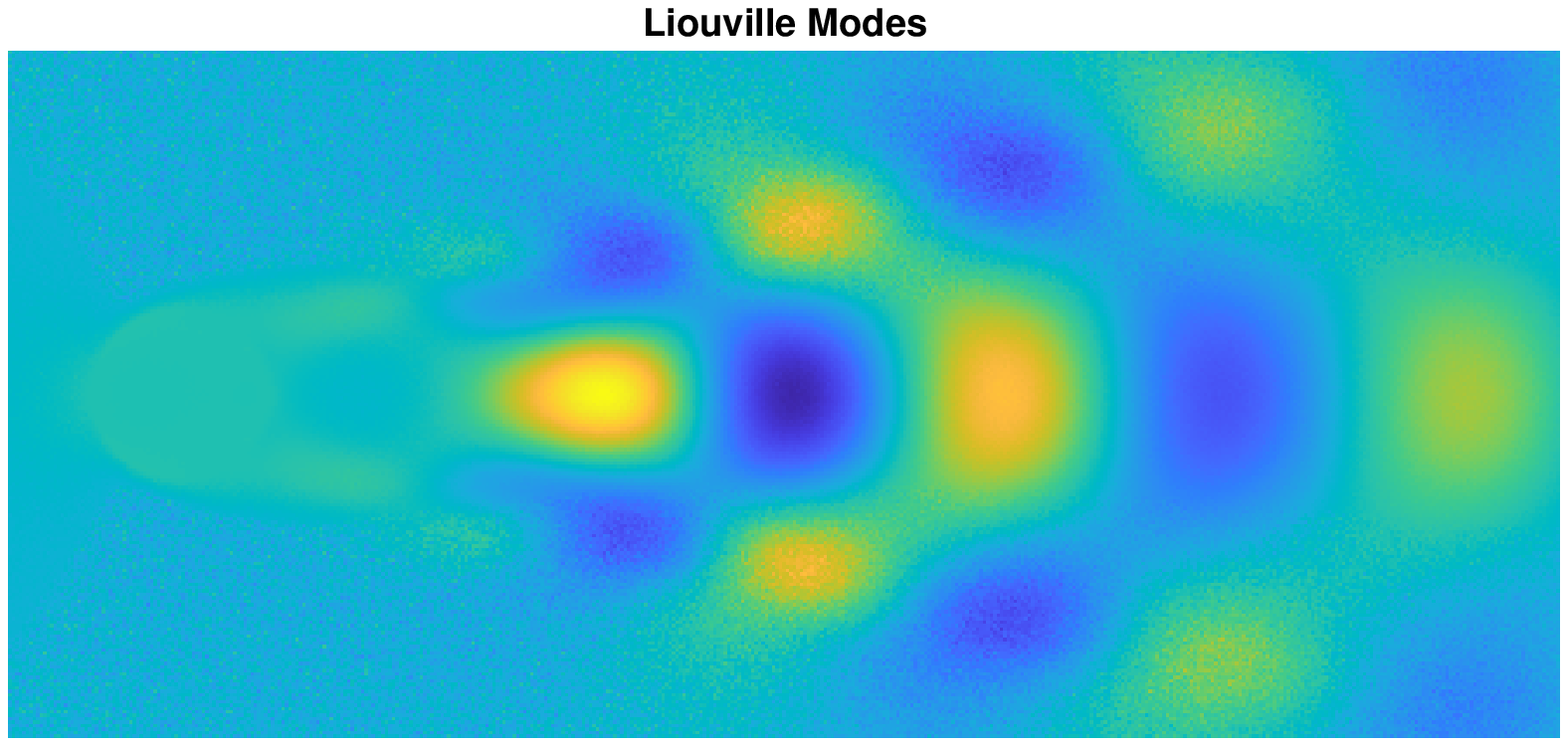}
    \includegraphics[width=0.48\textwidth]{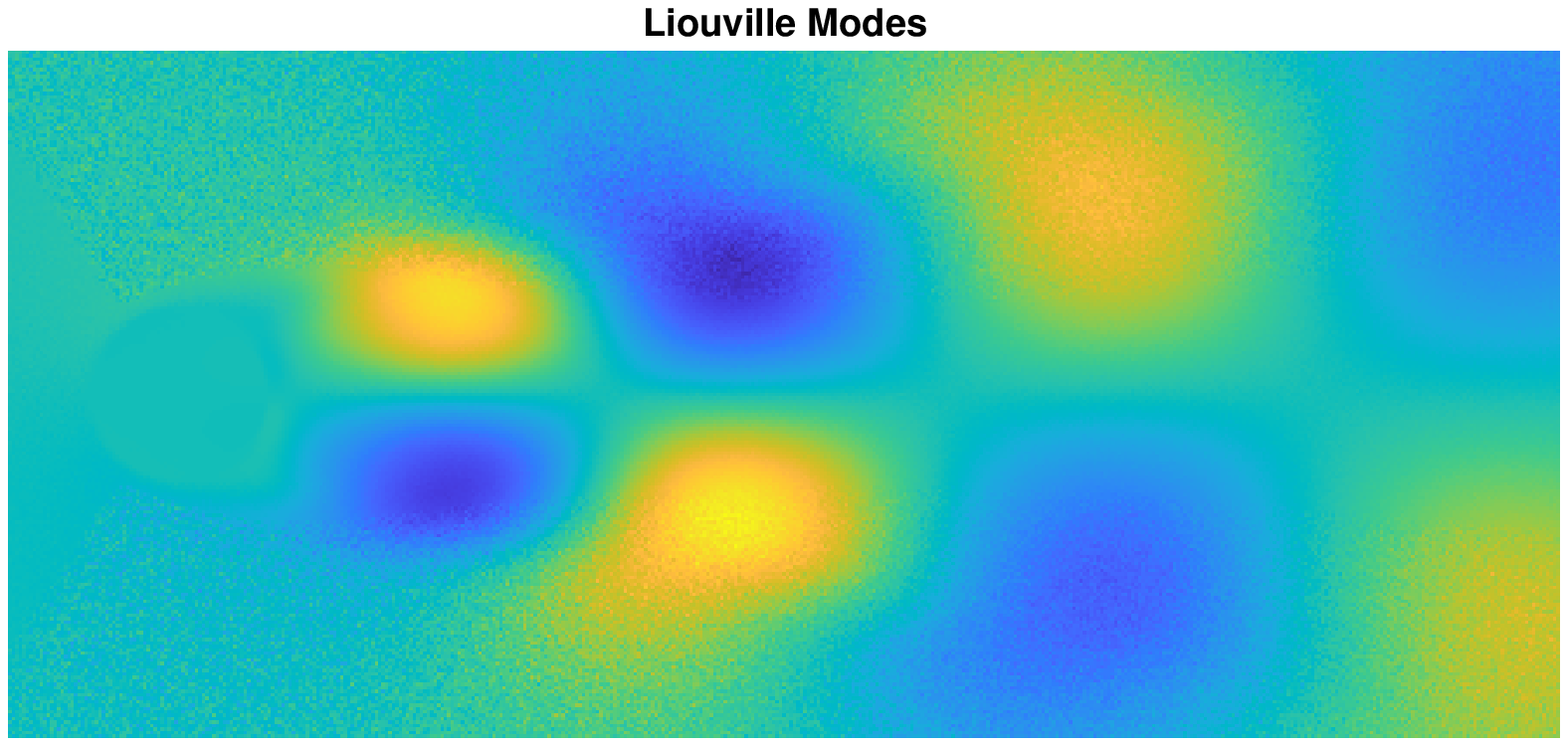}
    \caption{A selection of the real parts of approximate Liouville modes obtained using the exponential dot product kernel, where the domain corresponds to $\mu_1 = 1/1000$ and the range corresponds to $\mu_2 = 1/999$.}
\end{figure}

\begin{figure}
    \label{fig:reconstruction}
    \centering
    \includegraphics[width=0.48\textwidth]{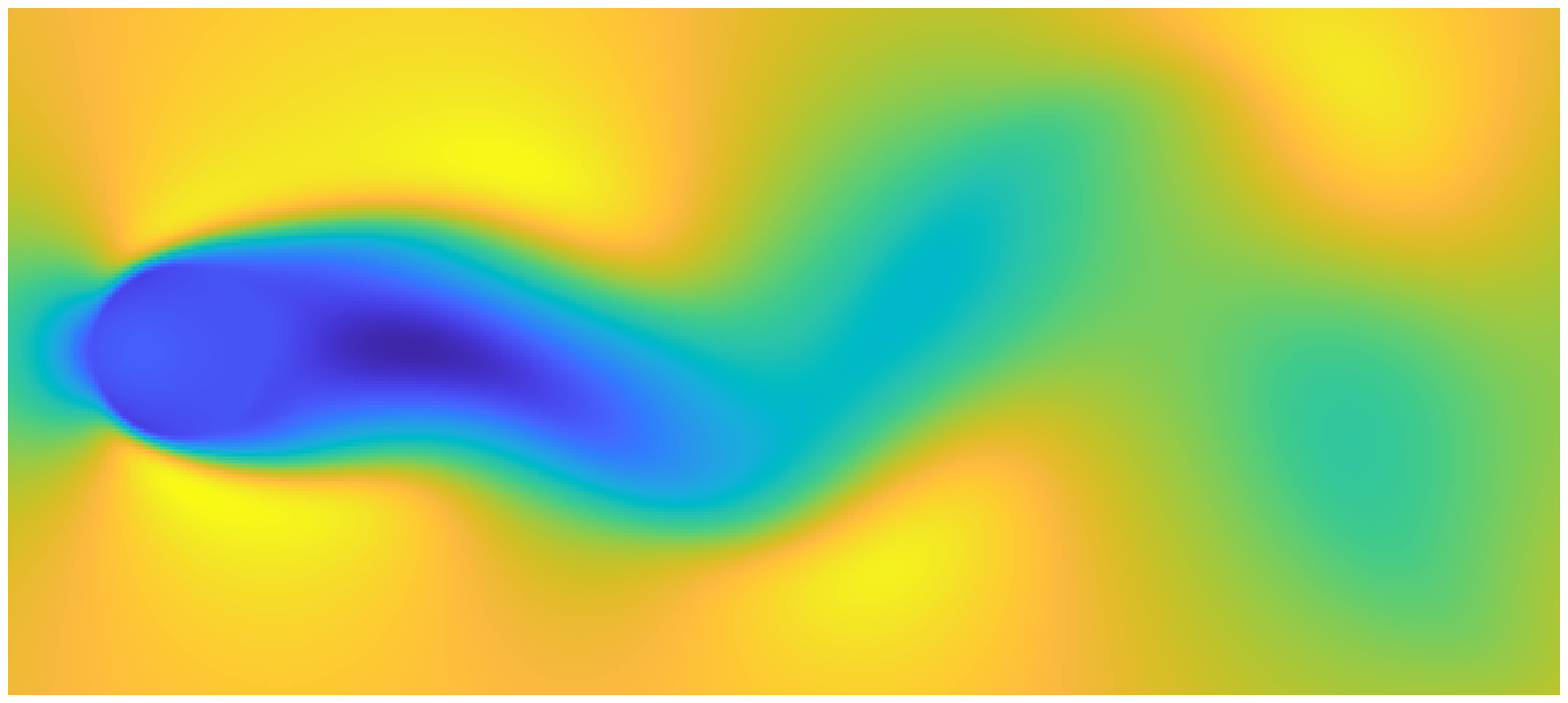}
    \includegraphics[width=0.48\textwidth]{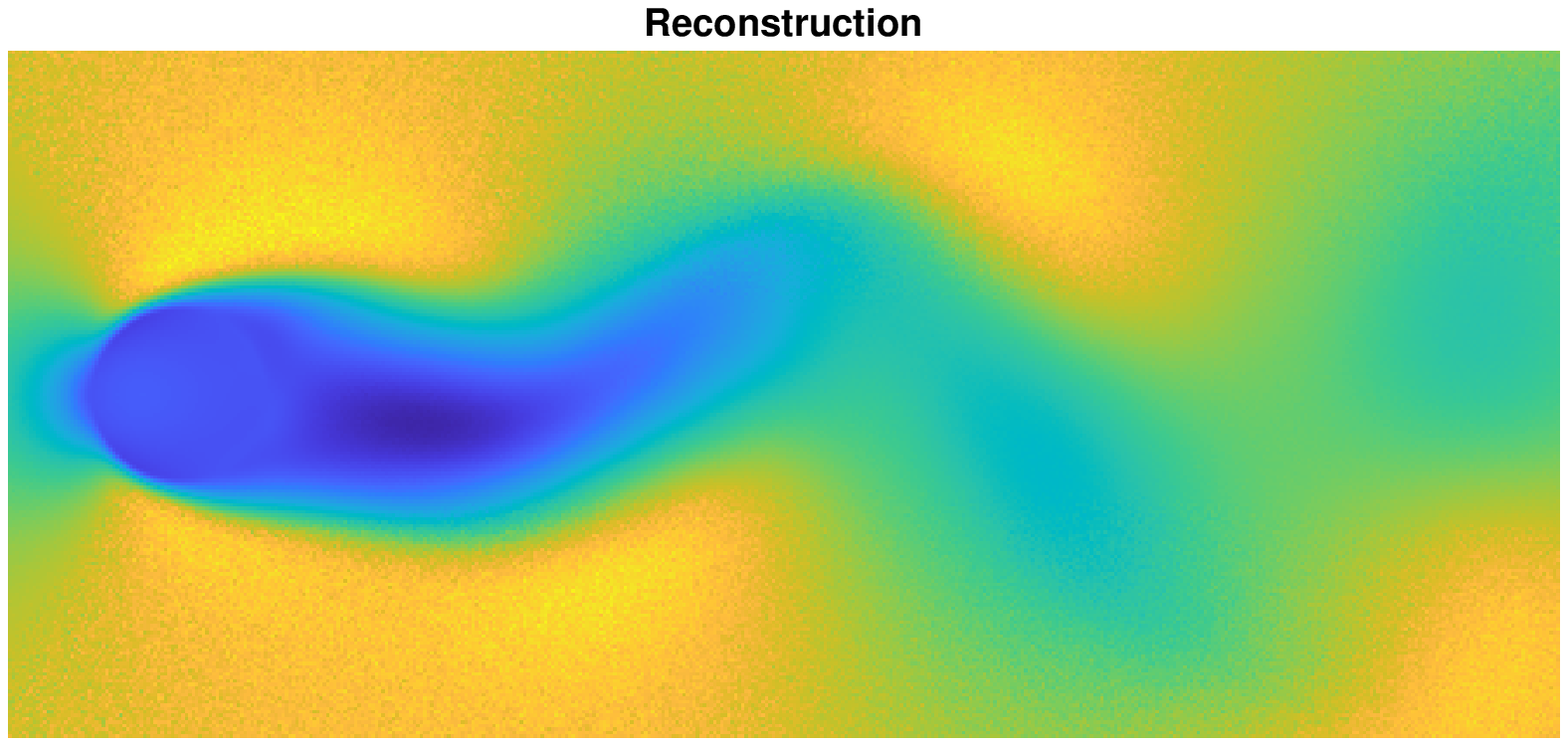}

    \includegraphics[width=0.48\textwidth]{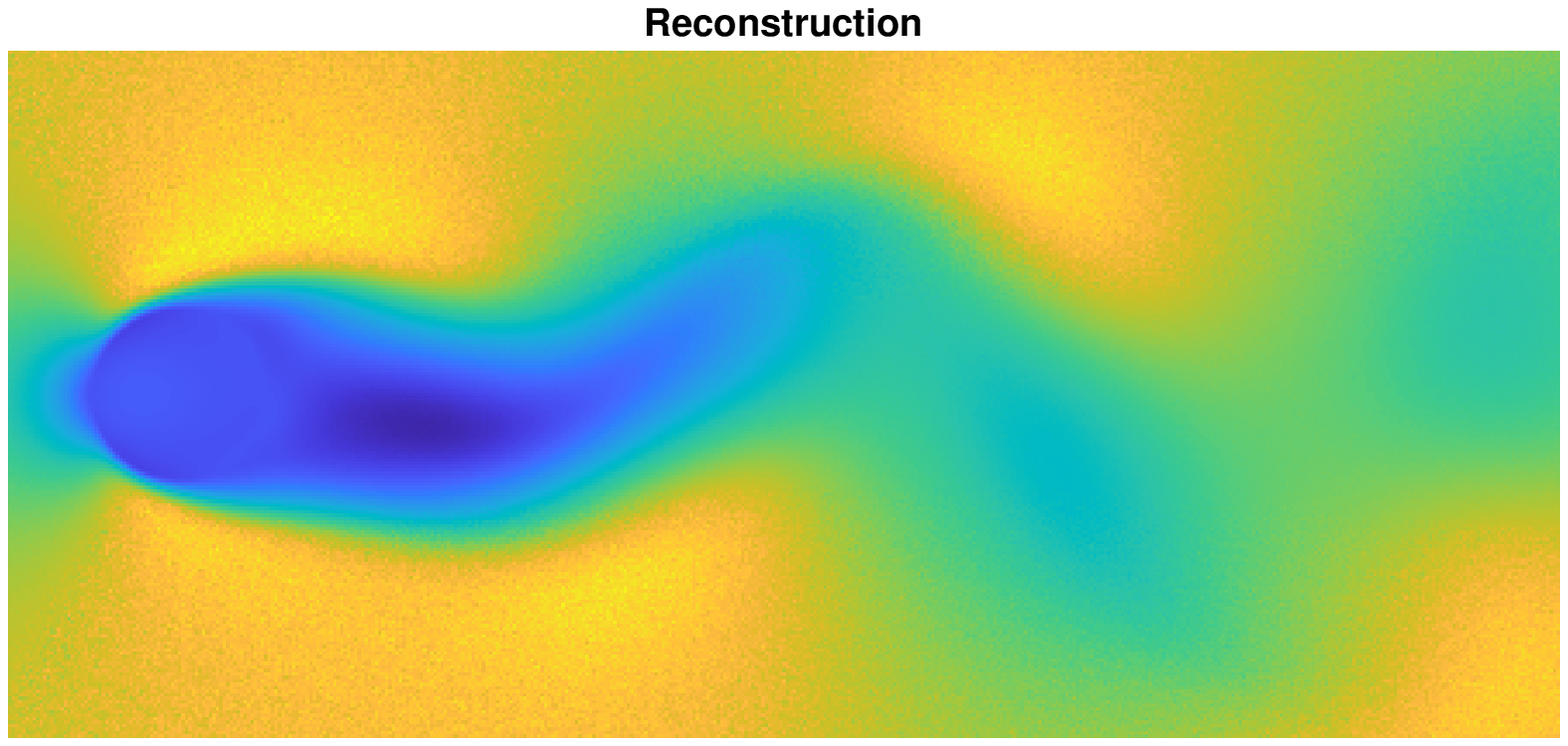}
    \includegraphics[width=0.48\textwidth]{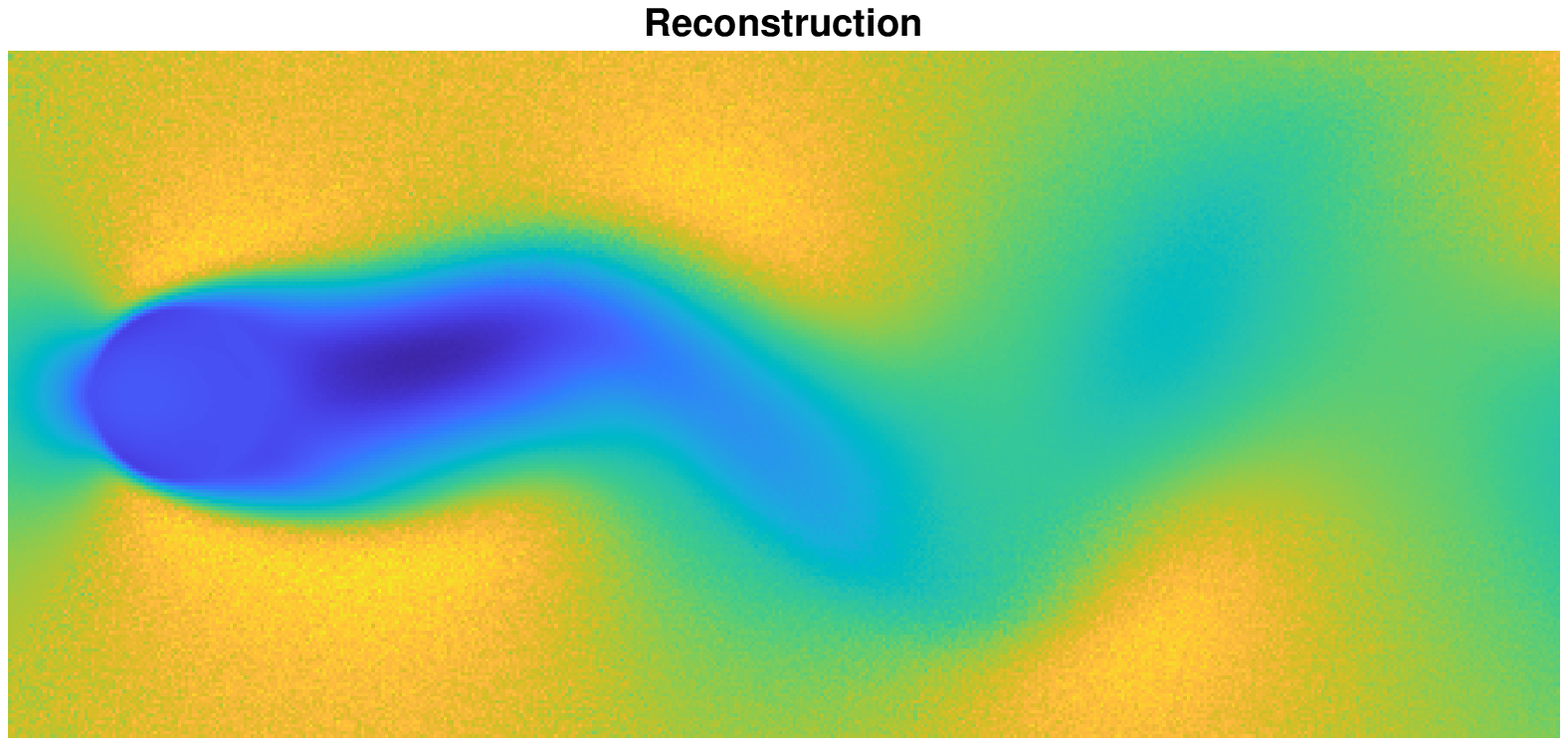}

    \includegraphics[width=0.48\textwidth]{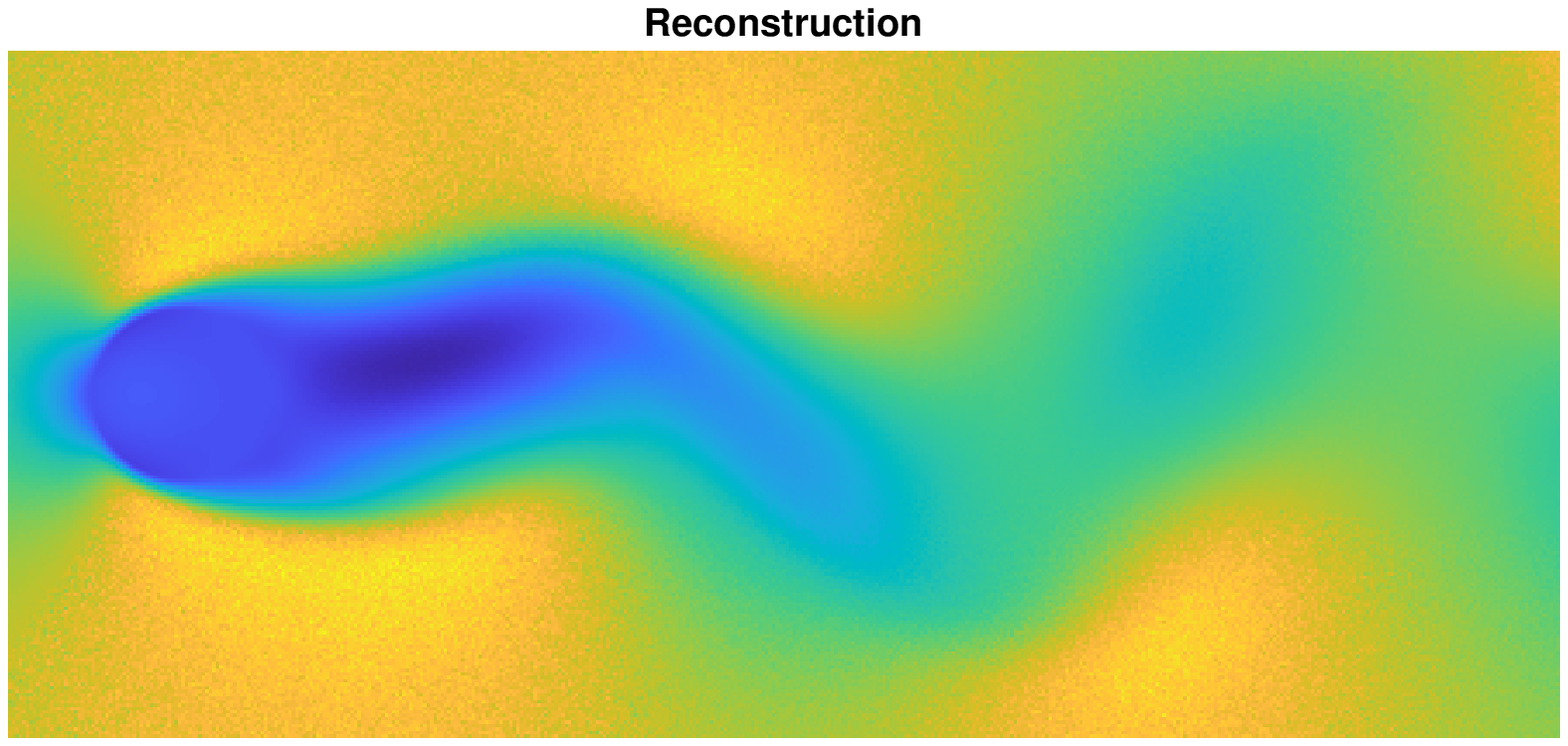}
    \includegraphics[width=0.48\textwidth]{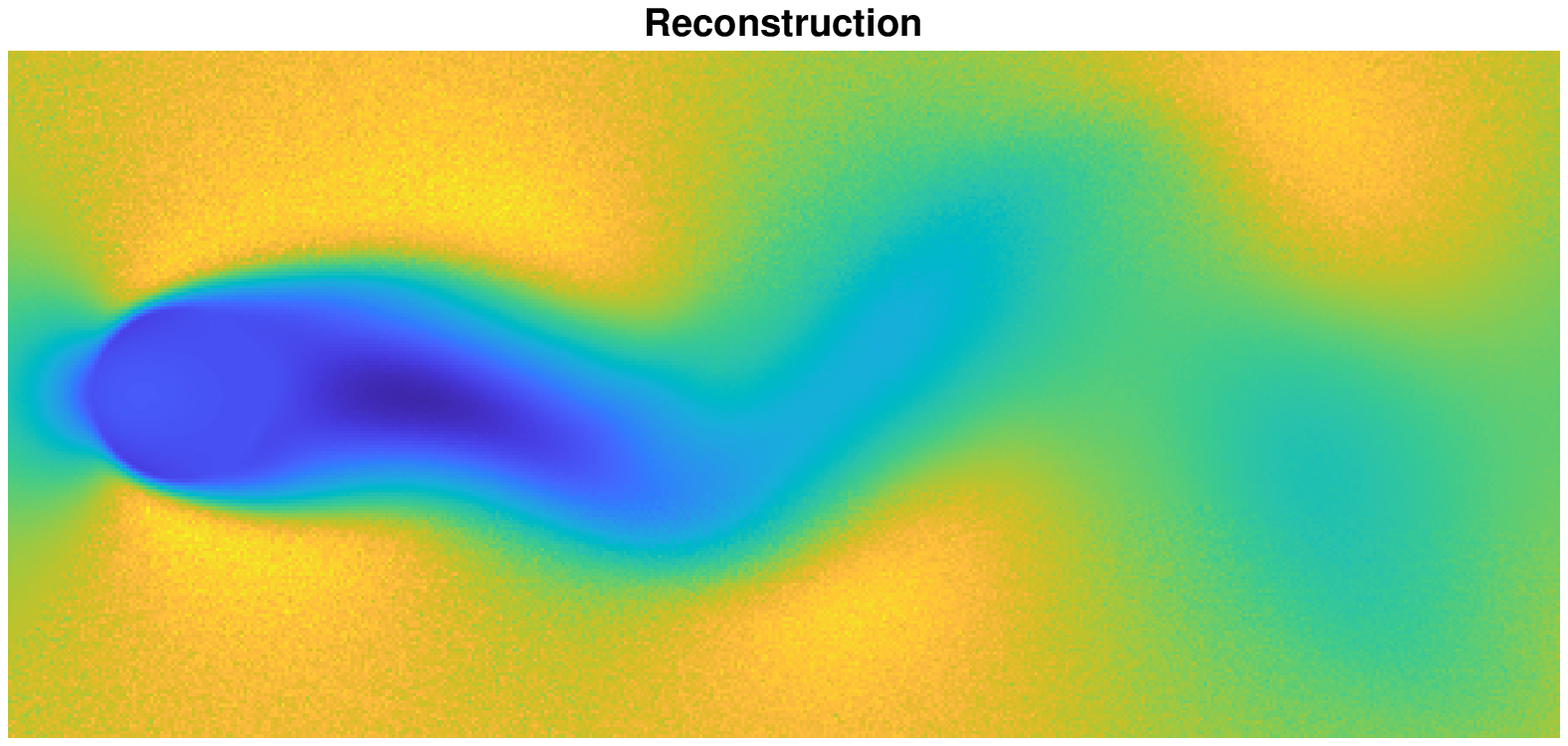}
    
    \includegraphics[width=0.48\textwidth]{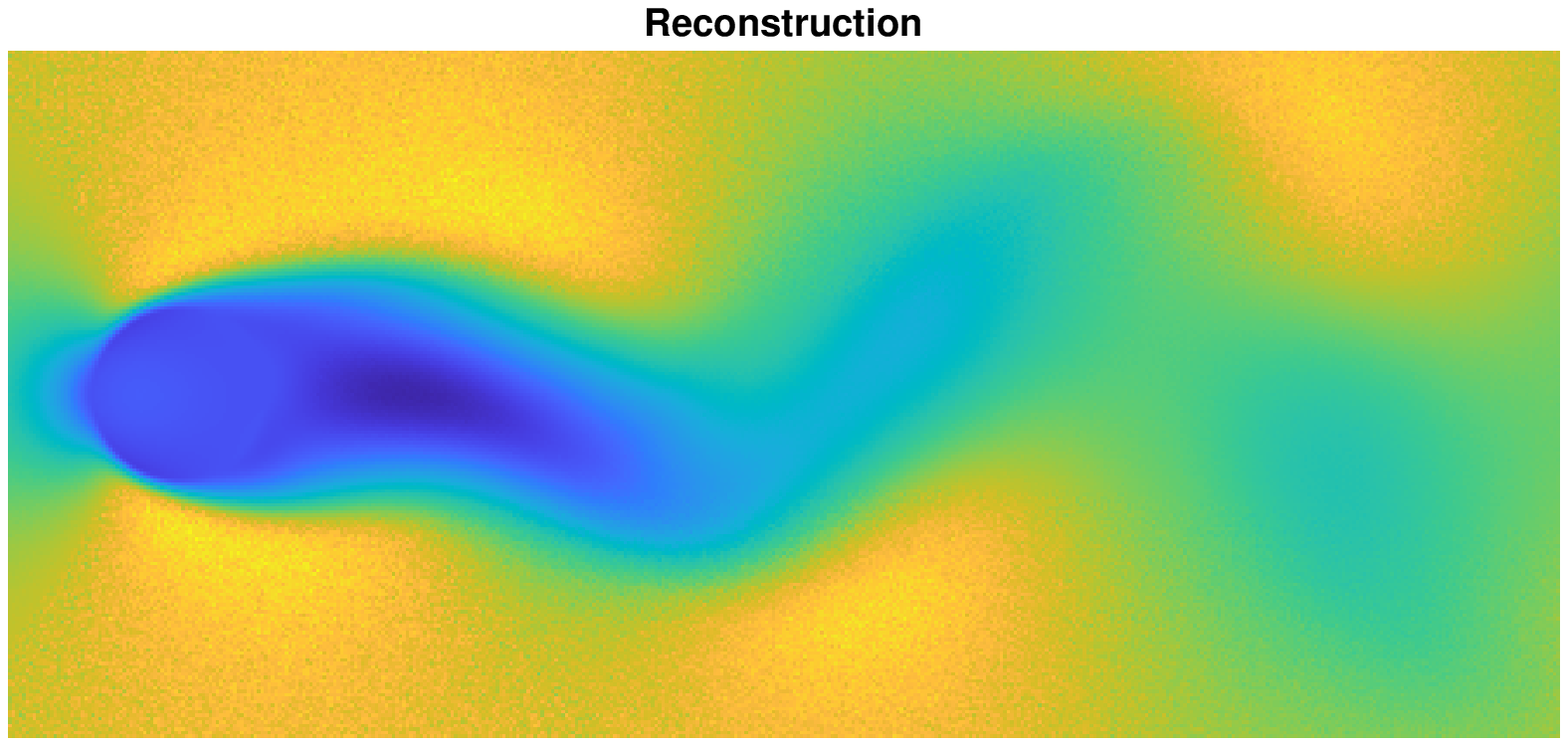}
    \includegraphics[width=0.48\textwidth]{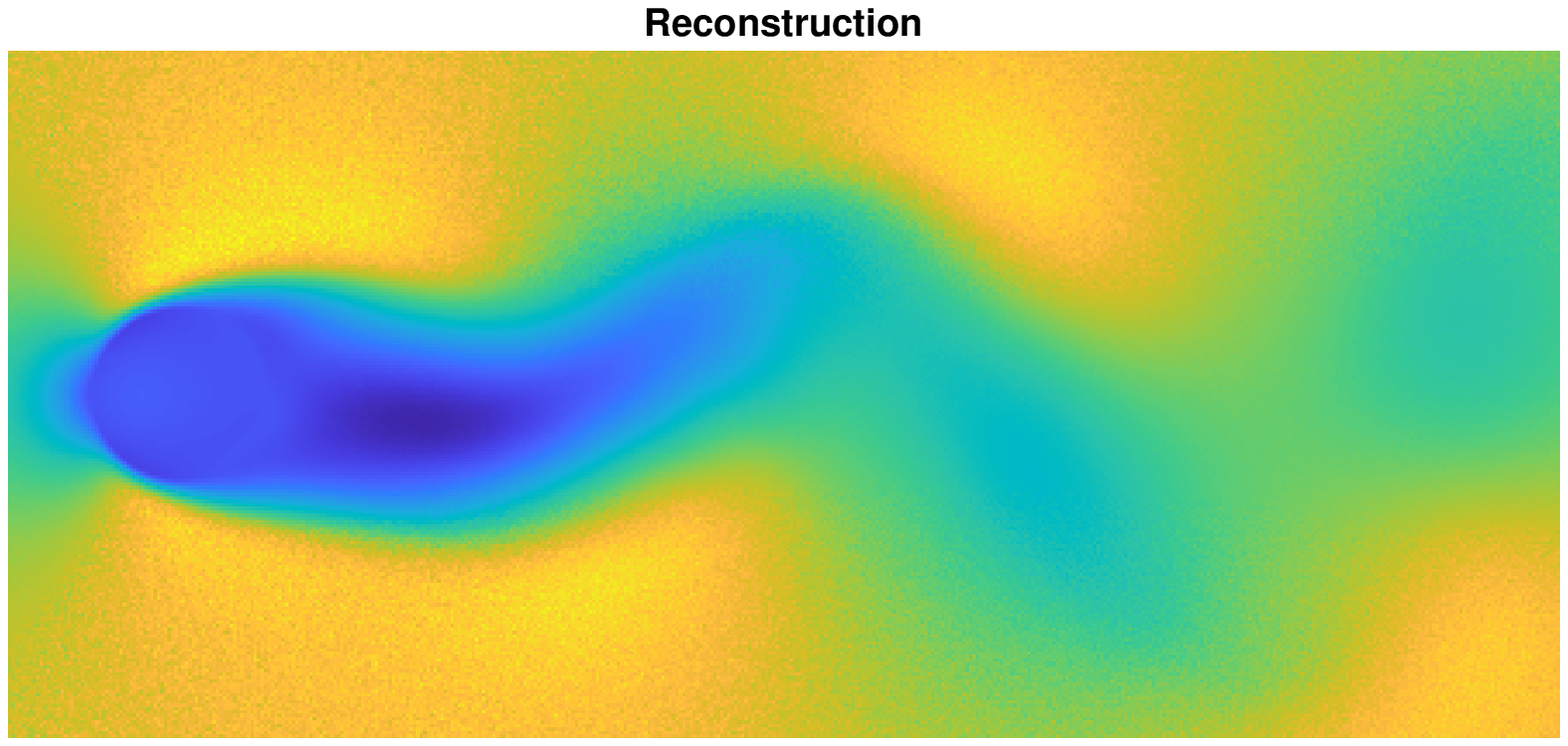}
    
        \caption{A selection of reconstructed snapshots for the cyllinder flow example. The first column from the top presents snapshots $1$, $21$, $41$, and $61$, and the second column presents $81$, $101$, $121$, and $141$.}
\end{figure}

\begin{figure}
    \label{fig:original}
    \centering
    \includegraphics[width=0.48\textwidth]{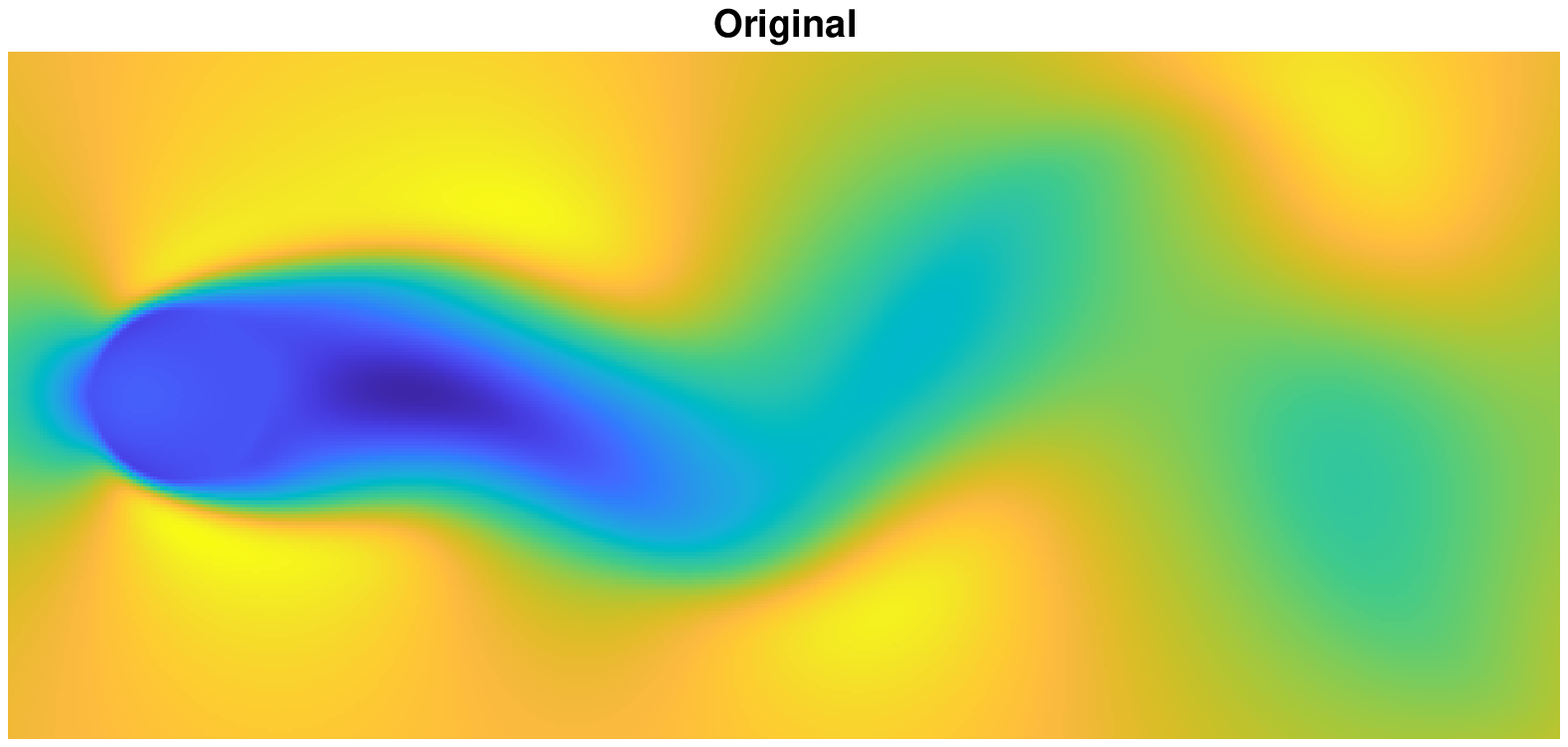}
    \includegraphics[width=0.48\textwidth]{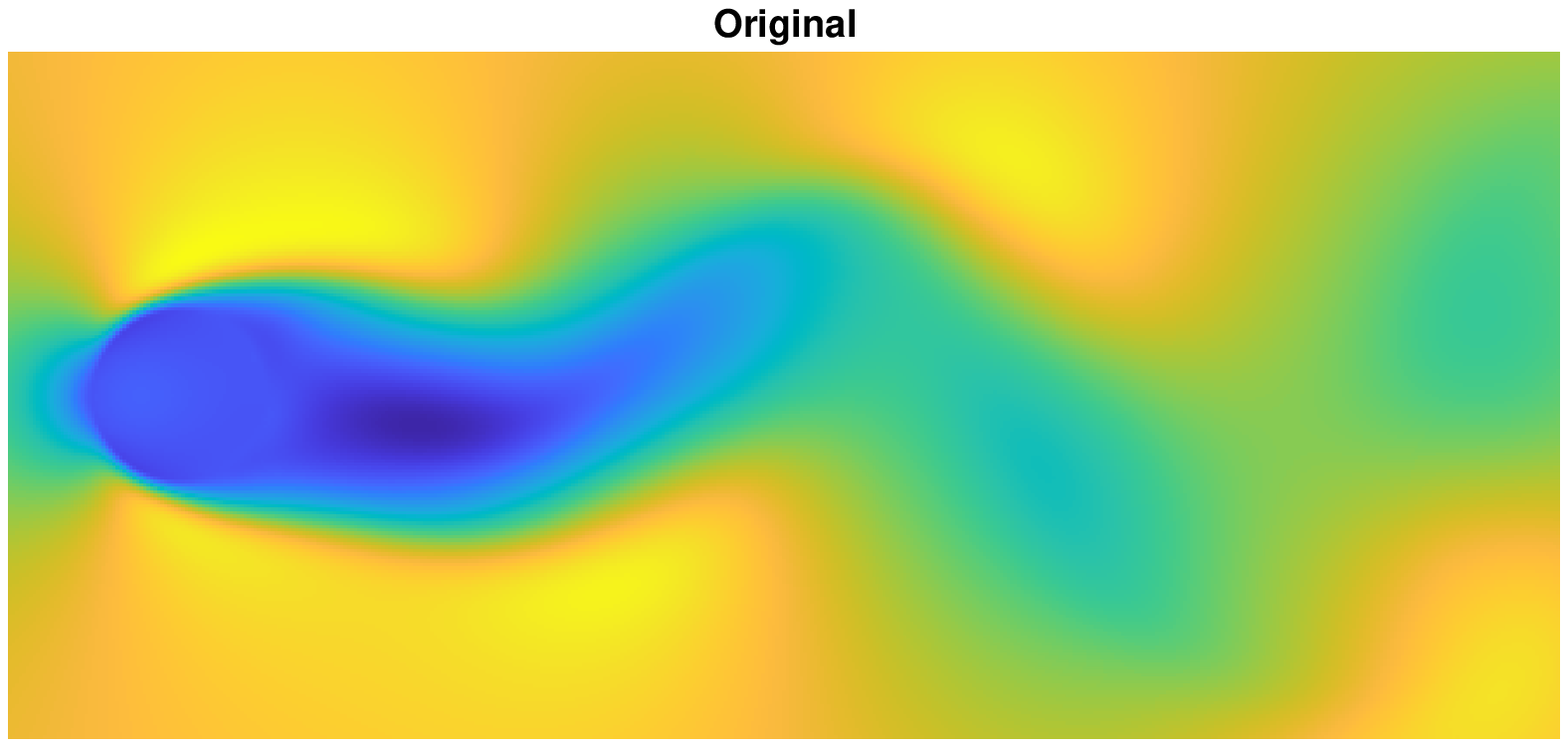}

    \includegraphics[width=0.48\textwidth]{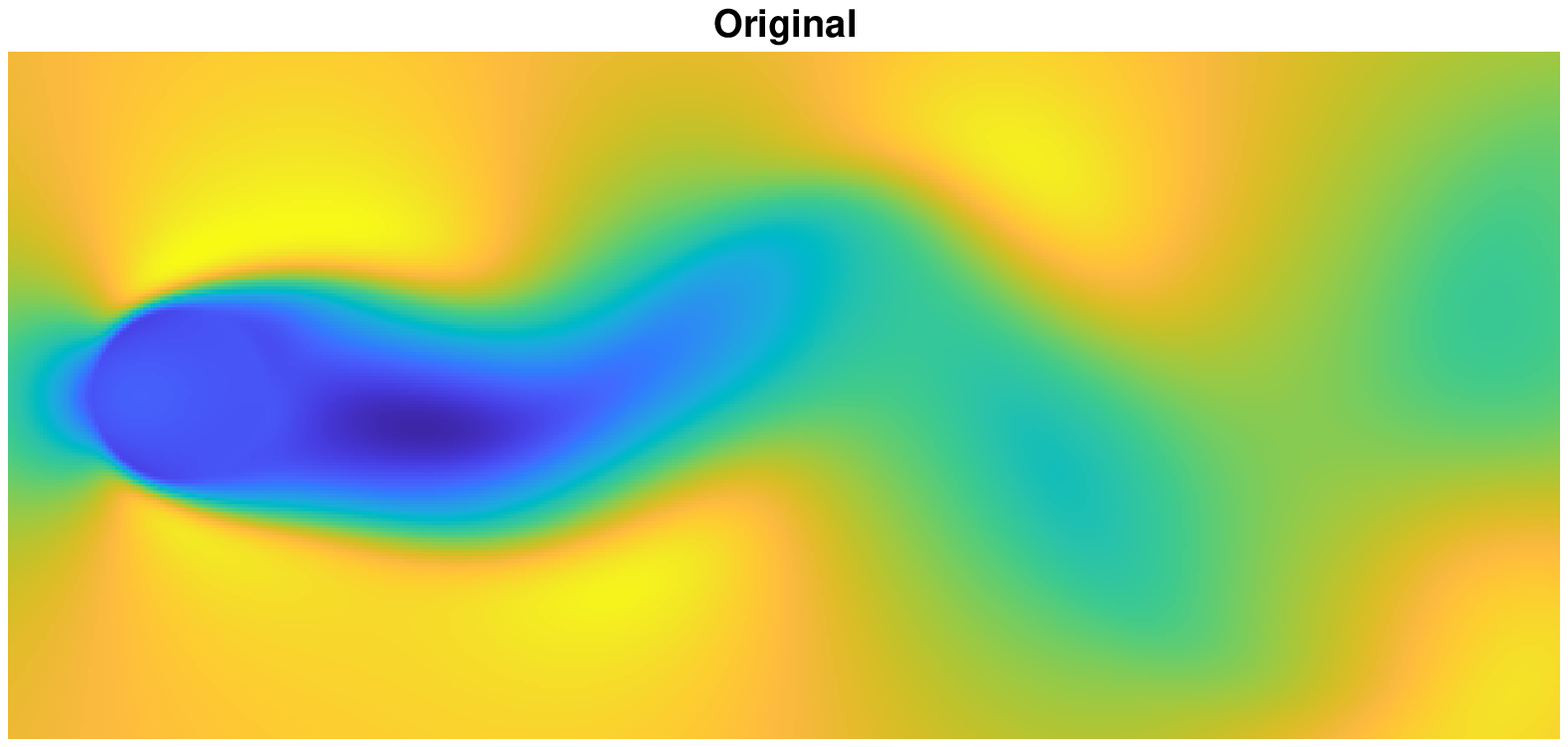}
    \includegraphics[width=0.48\textwidth]{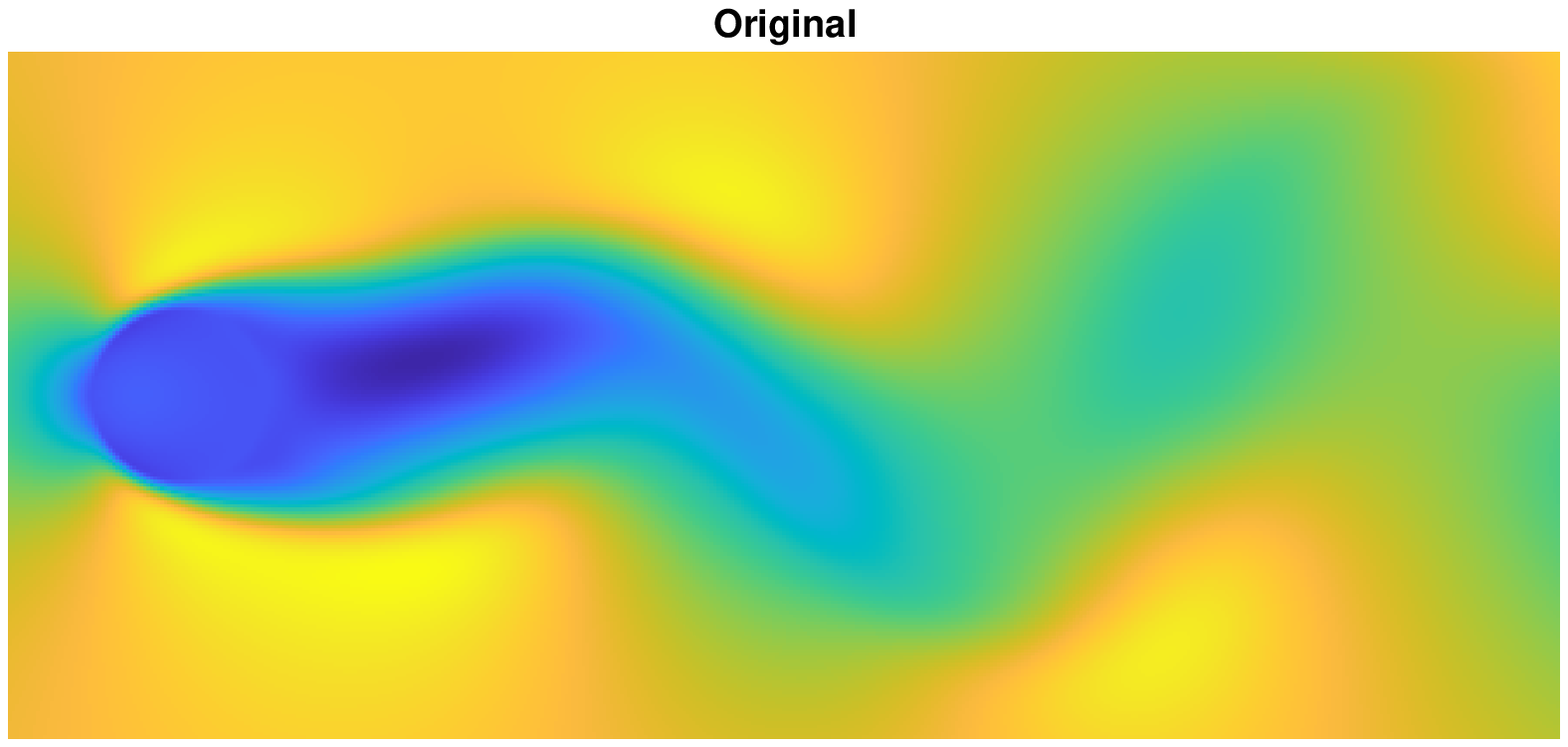}

    \includegraphics[width=0.48\textwidth]{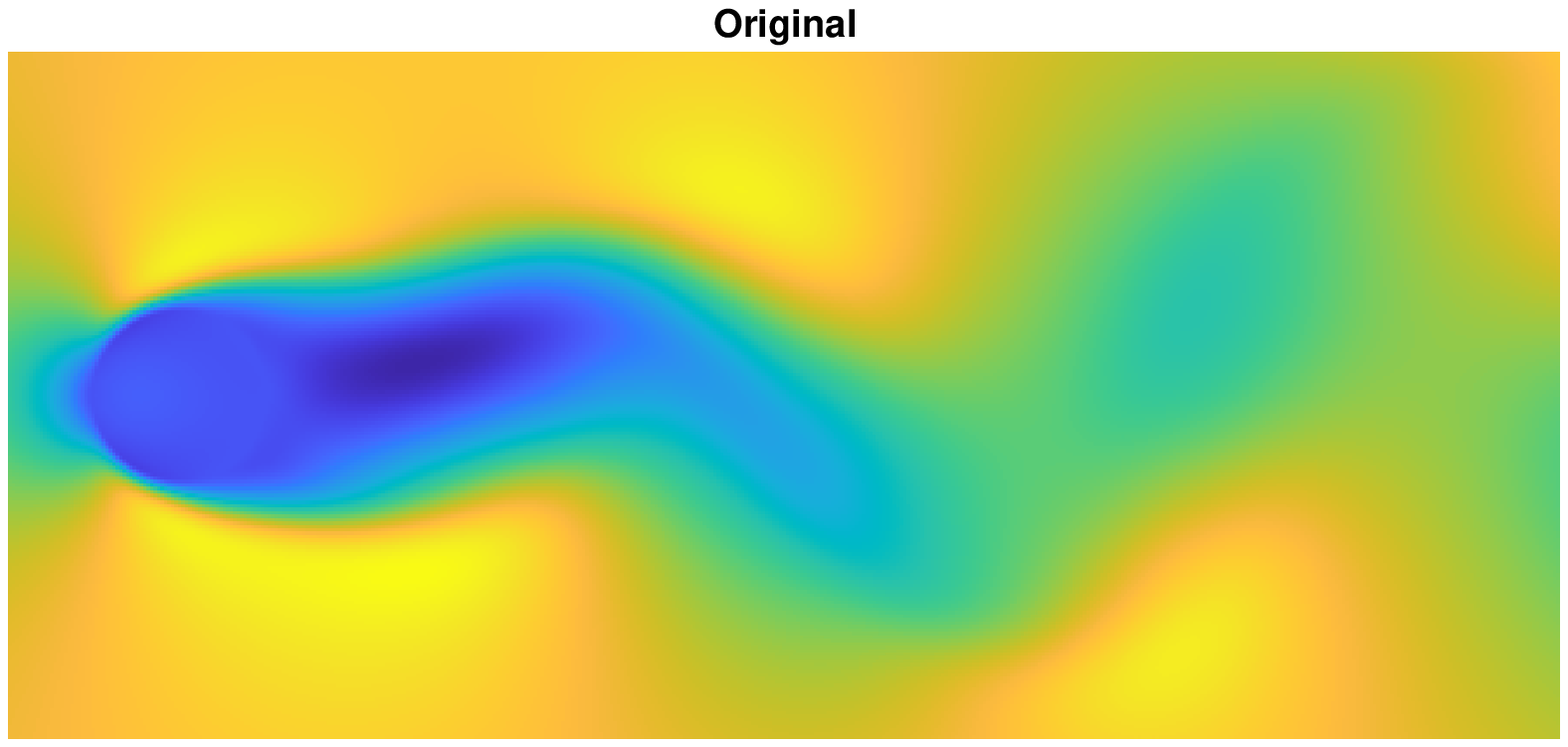}
    \includegraphics[width=0.48\textwidth]{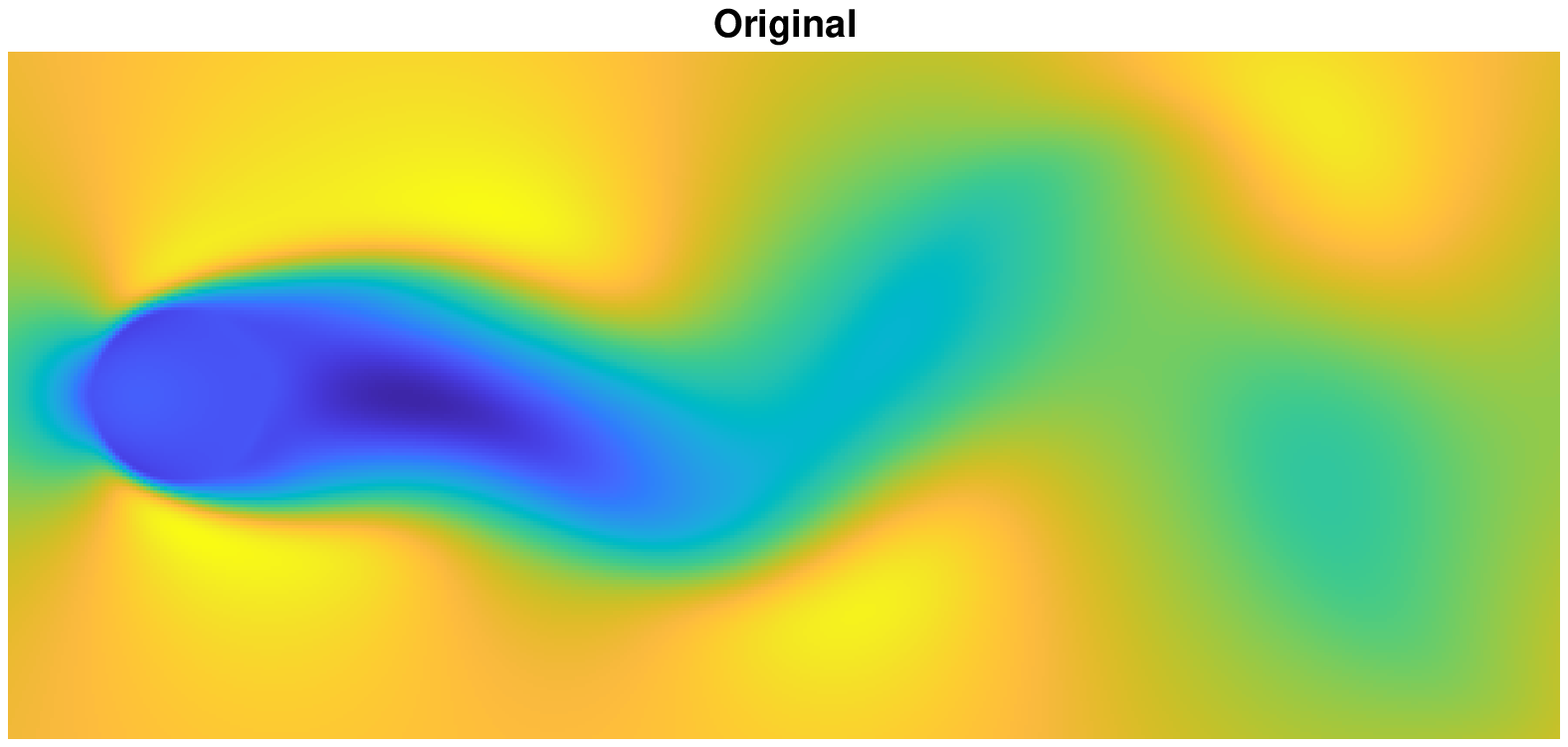}
    
    \includegraphics[width=0.48\textwidth]{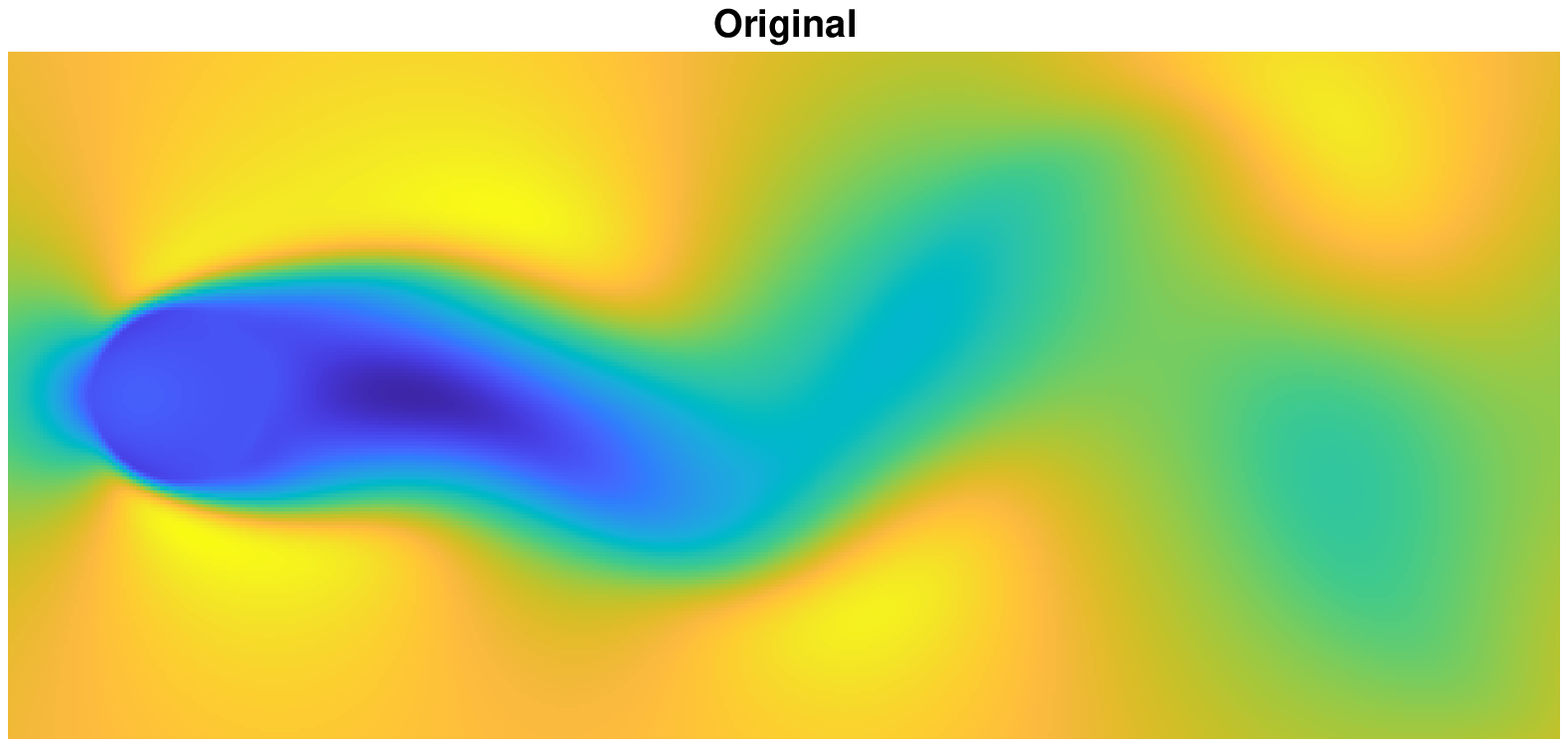}
    \includegraphics[width=0.48\textwidth]{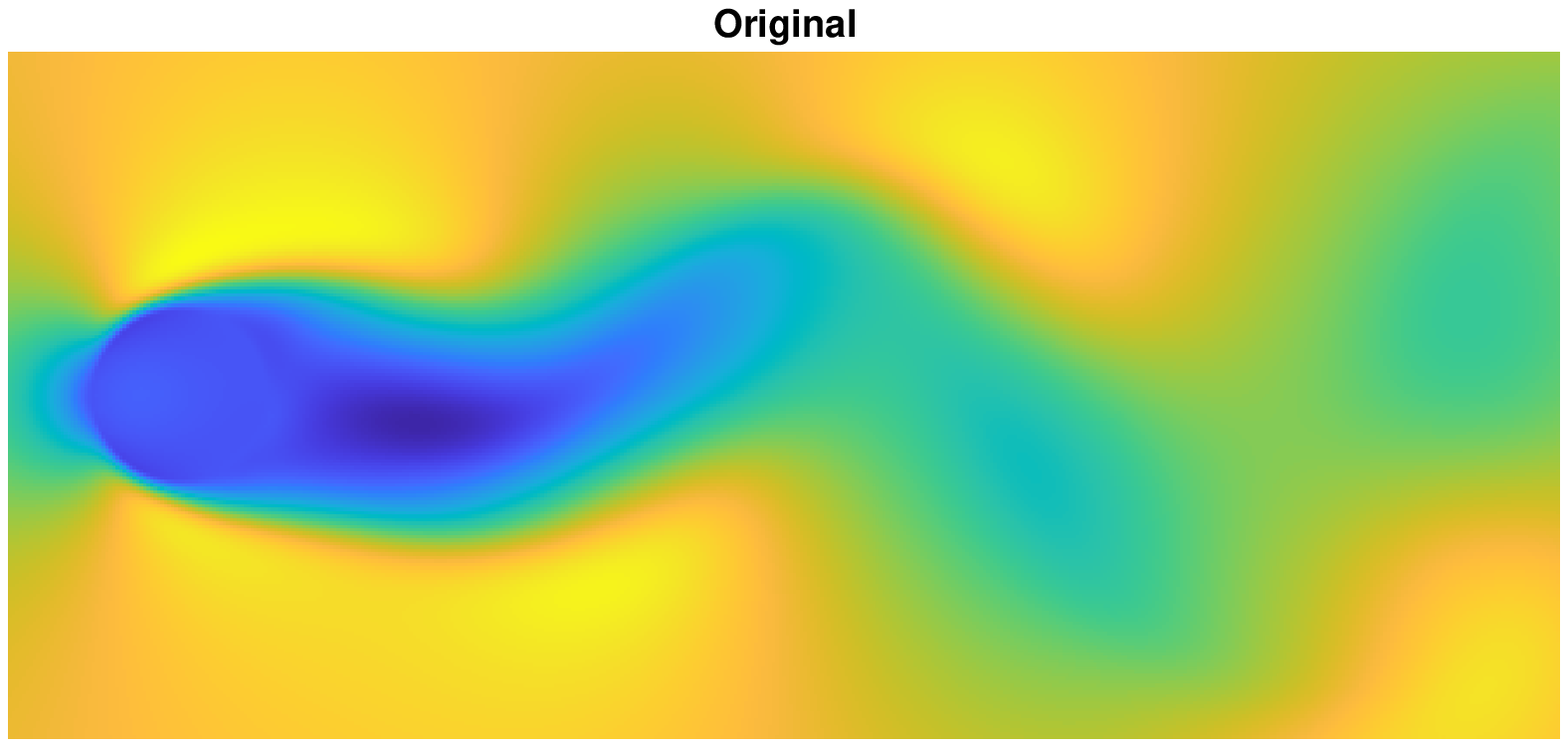}
    
        \caption{The original snapshots from the cyllinder flow data set in \cite{kutz2016dynamic}. The first column from the top presents snapshots $1$, $21$, $41$, and $61$, and the second column presents $81$, $101$, $121$, and $141$.}
\end{figure}

\section{Discussion}

The methods presented in this manuscript give two algorithms for performing a dynamic mode decomposition. Together with the compactness of the Liouville operators, the singular DMD approach guarantees the existence of dynamic modes and convergence through singular value decomposition of compact operators. Singular DMD is a general purpose approach to performing a dynamic mode decomposition for when the domain and range of the operators disagree. The major drawback of this approach is that even though it can guarantee the existence of dynamic modes, which cannot be done for eigenfunction methods, the reconstruction involves the solution of an initial value problem, which is technically more involved than the eigenfunction approach.

The second method adds an additional assumption to the problem, where the domain is assumed to be embedded in the range of the operator. These embeddings frequently occur in the study of RKHSs, where the adjustment of a parameter loosens the requirement on functions within that space. It was demonstrated that this embedding may be established for the exponential dot product kernel, and it also holds for the native spaces of Gaussian RBFs with differing parameters.

Convergence of these routines follow the proof found in \cite{rosenfeld2019dynamic}, which is a general purpose approach for showing convergence of operator level interpolants to the compact operators they are approximating. In particular, given an infinite collection of trajectories for a dynamical system, if the span of the occupation kernels form a dense subset of their respective Hilbert spaces, convergence of the overall algorithm is achieved.

The density of the occupation kernels corresponding to trajectories are easily established for Lipschitz continuous dynamics. This follows since, given any initial point, $x_0$ in $\mathbb{R}^n$, there is a $T_0$ such that the trajectory starting at $x_0$, $\gamma_{x_0}$, exists over the interval $[0,T_0]$. Consider the sequence of occupation kernels indexed by $\delta \in [0,T_0]$, $\Gamma_{\gamma_{x_0},\delta}(x) := \int_0^\delta K(x,\gamma_{x_0}(t))dt$. Then $\frac{1}{\delta} \Gamma_{\gamma_{x_0},\delta} \to K(x,x_0)$ in the Hilbert space norm. Hence, as $x_0$ was arbitrary, every kernel may be approximated by an occupation kernel corresponding to a trajectory, and since kernels are dense in $H$, so are these occupation kernels. Finally, if $H$ and $\tilde H$ are spaces of real analytic functions, the dynamics must also be real analytic by the same proof found in \cite{rosenfeld2019occupation}. Spaces of real analytic functions include the Gaussian RBF and the exponential dot product kernel space.

One interesting result of the structure of the finite rank approximation given in Section \ref{sec:eigenfunctiondmd} is that as $\mu_1 \to \mu_2$, the first two matrices cancel. The matrix computations then approach the computations in \cite{rosenfeld2019dynamic}. Hence, for close enough $\mu_1$ and $\mu_2$ the computations are computationally indistinguishable from \cite{rosenfeld2019dynamic} over a fixed compact set containing the trajectories.

Finally, it should be noted that this methodology is not restricted to spaces of analytic functions, but rather it can work for a large collection of pairs of spaces. As a rule, the range space should be less restrictive as to the collection of functions in that space than the domain space. With this in mind, for many of the cases where compact Liouville operators may be established, the domain will embed into the range. The complications arise in computing the first matrix in \eqref{eq:finiterankrep}, where the inner product of the occupation kernels for the domain are computed in the range space. Hence, the explicit description for spaces of real analytic functions help resolve that computation.

\section{Conclusion}
This manuscript presented a theoretical and algorithmic framework that achieves many long standing goals of dynamic mode decompositions. To wit, by selecting differing domains and ranges for the Liouville operators (sometimes Koopman generators), the resulting operators are compact. This comes at the sacrifice of eigenfunctions when the domain is not embedded in the range of the operator, but achieves well defined dynamic modes and convergence. Reconstruction can then be determined using typical numerical methods for initial value problems. However, in the case of an embedding between the spaces, an algorithm may be established to determine approximate eigenfunctions for the operators, resulting in a more typical DMD routine that also converges.

\bibliographystyle{siamplain}
\bibliography{references}
\end{document}